\documentclass[conference]{IEEEtran}

\ifCLASSINFOpdf
\else
\fi

\usepackage[cmex10]{amsmath}

\usepackage[lambda,advantage,operators,adversary,landau,probability,sets,ff,primitives,complexity,asymptotics,adversary,keys]{cryptocode}
\pcfixcleveref  % for compatibility with cryptocode
\usepackage{xspace}
\usepackage[nolist]{acronym}
\usepackage{bm}
\usepackage[framemethod=TikZ]{mdframed}
\mdfsetup{
    roundcorner=5pt,
    backgroundcolor=gray!15,
    innertopmargin=5,
    innerleftmargin=5,
    innerbottommargin=5,
    innerrightmargin=5,
    skipabove=5,
    skipbelow=1,
    linewidth=1,
    nobreak=true
}
\usepackage{stmaryrd}
\usepackage{amsthm}
\usepackage{graphicx}
\usepackage[caption=false,font=footnotesize]{subfig}
\let\labelindent\relax  % for compatibility with enumitem
\usepackage{enumitem}
\usepackage[hidelinks]{hyperref}
\usepackage[capitalise, nameinlink, noabbrev]{cleveref}
\newtheorem{theorem}{Theorem}
\newtheorem{definition}{Definition}

\let\originalpcalgostyle\pcalgostyle
\renewcommand*{\pcalgostyle}[1]{\originalpcalgostyle{#1}\xspace}
\let\originalpckeystyle\pckeystyle
\renewcommand*{\pckeystyle}[1]{\originalpckeystyle{#1}\xspace}
\let\originalpcadvstyle\pcadvstyle
\renewcommand*{\pcadvstyle}[1]{\originalpcadvstyle{#1}\xspace}
\let\originalpcsetstyle\pcsetstyle
\renewcommand*{\pcsetstyle}[1]{\originalpcsetstyle{#1}\xspace}
\let\originalsecpar\secpar
\renewcommand*{\secpar}{\originalsecpar\xspace}

\renewcommand*{\verify}{\pcalgostyle{Verify}}
\renewcommand*{\kgen}{\pcalgostyle{KeyGen}}

\newcommand*{\slice}{\mathbin{:}}
\newcommand*{\domain}{\ensuremath{\mathcal{D}}\xspace}
\newcommand*{\relation}{\ensuremath{R}\xspace}
\newcommand*{\relationfamily}{\ensuremath{\mathcal{R}_\secpar}\xspace}
\newcommand*{\transcript}{\textsf{tr}\xspace}
\newcommand*{\vdv}{\pcadvstyle{V}}
\newcommand*{\view}{\textsf{View}\xspace}

\newcommand*{\ck}{\pckeystyle{ck}}
\newcommand*{\ek}{\pckeystyle{ek}}

\newcommand*{\ccp}{\pcalgostyle{cCP}}

\newcommand{\cplink}{\ensuremath{\pcalgostyle{CP}_\text{link}}\xspace}
\newcommand*{\prove}{\pcalgostyle{Prove}}
\newcommand*{\protocol}{\ensuremath{\Pi}\xspace}
\newcommand*{\setup}{\pcalgostyle{Setup}}
\newcommand*{\comm}{\pcalgostyle{Comm}}
\newcommand*{\commit}{\pcalgostyle{Commit}}
\newcommand*{\vercommit}{\pcalgostyle{VerCommit}}

\renewcommand{\IEEEiedlistdecl}{\setlength{\labelsep}{0.5em}}

\begin{document}

% define acronyms
\begin{acronym}[zk-SNARK]
    \acro{ccnizk}[cc-NIZK]                  {commit-carrying NIZK}
    \acro{ccsnark}[cc-SNARK]                {commit-carrying zk-SNARK}
    \acro{cpsnark}[CP-SNARK]                {commit-and-prove zk-SNARK}
    \acro{cts}[CtS]                         {Commit-then-Share}
    \acro{cpnizk}[CP-NIZK]                  {commit-and-prove NIZK}
    \acro{CRS}                              {common reference string}
    \acro{DAG}                              {directed acyclic graph}
    \acro{IOP}                              {interactive oracle proof}
    \acro{IPA}                              {inner product argument}
    \acro{MAC}                              {message authentication code}
    \acro{MPC}                              {secure multiparty computation}
        \acrodefindefinite{MPC}             {an}{a}
    \acro{nizk}[NIZK]                       {non-interactive zero-knowledge}
    \acro{PA-MPC}                           {publicly auditable MPC}
    \acro{ppt}[p.p.t.\@]                    {probabilistic polynomial time}
    \acro{R1CS}                             {rank-1 constraint system}
    \acro{stc}[StC]                         {Share-then-Commit}
    \acro{VC}                               {verifiable computing}
    \acro{VPPC}                             {verifiable privacy-preserving computation}
    \acro{ZKP}                              {zero-knowledge proof}
    \acro{zksnark}[zk-SNARK]                {zero-knowledge succinct non-interactive argument of knowledge}
        \acrodefplural{zksnark}[zk-SNARKs]  {zero-knowledge succinct non-interactive arguments of knowledge}
\end{acronym}

\title{Collaborative CP-NIZKs: Modular, Composable Proofs for Distributed Secrets}


\author{\IEEEauthorblockN{Mohammed Alghazwi}
\IEEEauthorblockA{University of Groningen\\
m.a.alghazwi@rug.nl}
\and
\IEEEauthorblockN{Tariq Bontekoe}
\IEEEauthorblockA{University of Groningen\\
t.h.bontekoe@rug.nl}
\and
\IEEEauthorblockN{Leon Visscher}
\IEEEauthorblockA{University of Groningen\\
l.visscher.2@student.rug.nl}
\and
\IEEEauthorblockN{Fatih Turkmen}
\IEEEauthorblockA{University of Groningen\\
f.turkmen@rug.nl}
}

% \IEEEoverridecommandlockouts
% \makeatletter\def\@IEEEpubidpullup{6.5\baselineskip}\makeatother
% \IEEEpubid{\parbox{\columnwidth}{
%     % Network and Distributed System Security (NDSS) Symposium 2025\\
%     % 23 February - 28 February 2025, San Diego, CA, USA\\
%     % ISBN 1-891562-93-2\\
%     % https://dx.doi.org/10.14722/ndss.2025.23xxx\\
%     % www.ndss-symposium.org
% }
% \hspace{\columnsep}\makebox[\columnwidth]{}}

\maketitle

\begin{abstract}
    \Ac{nizk} proofs of knowledge have proven to be highly relevant for securely realizing a wide array of applications that rely on both \emph{privacy} and \emph{correctness}.
    They enable a prover to convince any party of the correctness of a public statement for a \emph{secret witness}.
    However, most \acp{nizk} do not natively support proving knowledge of a secret witness that is distributed over multiple provers.
    Previously, collaborative proofs~\cite{ozdemirExperimentingCollaborativeZkSNARKs2022} have been proposed to overcome this limitation.
    We investigate the notion of composability in this setting, following the Commit-and-Prove design of LegoSNARK~\cite{campanelliLegoSNARKModularDesign2019}.
    Composability allows users to combine different, specialized \acp{nizk} (e.g., one arithmetic circuit, one boolean circuit, and one for range proofs) with the aim of reducing the prove generation time.
    Moreover, it opens the door to efficient realizations of many applications in the collaborative setting such as mutually exclusive prover groups, combining collaborative and single-party proofs and efficiently implementing \ac{PA-MPC}.

    We present the first, general definition for collaborative \ac{cpnizk} proofs of knowledge and construct distributed protocols to enable their realization. We implement our protocols for two commonly used \acp{nizk}, Groth16 and Bulletproofs, and evaluate their practicality in a variety of computational settings. Our findings indicate that composability adds only minor overhead, especially for large circuits. We experimented with our construction in an application setting, and when compared to prior works, our protocols reduce latency by \mbox{$18$--$55\times$} while requiring only a fraction ($0.2\%$) of the communication.
\end{abstract}

% content of the paper
\section{Introduction} \label{sec:intro}
\Iac{ZKP}~\cite{goldwasserKnowledgeComplexityInteractive1989} is a cryptographic construction, enabling a prover to convince a verifier of the truth of a statement, without revealing anything other than its validity.
Its zero-knowledge property ensures that \emph{privacy} of the prover's secrets is preserved, whilst its soundness property guarantees that the statement is valid (i.e., we have \emph{correctness}) and that the prover did not cheat.
In this work, we focus on a broad class of \acp{ZKP} known as \ac{nizk} proofs of knowledge.
They are \emph{publicly verifiable} due to their non-interactiveness, i.e., the verifier is not involved in proof generation.
Being a proof of knowledge entails that the prover not only convinces the verifier of the validity of a statement, but also that it \emph{knows} the secret data, or \emph{witness}, for which the statement is valid.
Note, that \acsp{zksnark} are a subset of \ac{nizk} proofs of knowledge, i.e., they form the subset of succinct proofs.
% Zero-knowledgeness ensures that this witness is not revealed to any other party.

In recent years, many different constructions for \ac{nizk} proofs have been proposed for generic statements, e.g., those expressible as an arithmetic circuit. Each of these constructions comes with its own advantages, e.g., no trusted setup, constant proof size, efficient verification, or relying only on standard cryptographic assumptions.
However, these generic schemes have a few major limitations~\cite{ozdemirExperimentingCollaborativeZkSNARKs2022}:
\begin{enumerate}[leftmargin=*]
    \item \emph{Proof generation is costly;}
    \item \emph{The witness should be known by the prover (a single party).}
\end{enumerate}
The former refers to the time/space it takes to generate proofs. The latter instead refers to the limitation that the generation of a proof where the witness is \textit{distributed} over multiple parties is not directly supported.

At the same time, there is an array of specialized \ac{nizk} schemes, e.g.,~\cite{cormodePracticalVerifiedComputation2012,wahbyDoublyEfficientZkSNARKsTrusted2018,campanelliLegoSNARKModularDesign2019}, that enable provers to generate proofs for specific types of statements more efficiently.
This brings us to our third and final major limitation: 
\begin{enumerate}[leftmargin=*]
\setcounter{enumi}{2}
\item \emph{Proof schemes are not efficiently composable.}
\end{enumerate}
Thus, we cannot natively benefit from the distinguishing advantages of different schemes, by creating a proof for a single statement using a composition of different schemes.

\textbf{Collaborative proofs.}
The first two limitations could be tackled by considering the notion of a distributed prover (where all provers have access to the full witness).
While this would not directly improve efficiency, it would enable the outsourcing of proof generation in a privacy-preserving manner with the additional benefit of load distribution~\cite{wu_dizk_2018,garg_zksaas_2023,chiesa_eos_2023}.

Even more interesting is the case where the secret witness is also distributed amongst multiple provers, i.e., no prover knows the full witness.
This is particularly useful when several parties wish to prove a statement about secret data distributed amongst them without revealing their part of the witness to each other or the verifier.
Recently, a generic definition for collaborative \acsp{zksnark}~\cite{ozdemirExperimentingCollaborativeZkSNARKs2022} was proposed as a solution for this setting.
Notably, this approach can also be used to generically construct \ac{PA-MPC}~\cite{baum_publicly_2014}.
However, we note that for \ac{PA-MPC}, the provers need to open a large number of commitments inside a proof circuit, which is computationally inefficient.
As we will show in this work, our proposal for collaborative \acp{cpnizk} can be used to make this construction scalable.

\textbf{Composability.}
In accordance with the third limitation, we observe that proof statements are often composed of several substatements, for example by conjunction (see \cref{subsec:composition}).
One substatement may be more efficiently proved in one scheme (e.g., arithmetic circuits \ac{nizk}) and a second more efficiently in another scheme (e.g., specialized or boolean circuit \ac{nizk}).
To benefit from both the efficiency of specialized \acp{nizk} and the flexibility of general-purpose \acp{nizk}, it would be ideal to combine multiple schemes to gain efficiency.
At first glance, it may seem sufficient to naively generate a proof for each substatement and have the verifier check all proofs separately.
However, since the verifier does not know the witness, which is often (partially) shared amongst substatements, it cannot verify whether the proofs are for the same witness, i.e., correctness is no longer guaranteed.

The approach where multiple subproofs are combined into a single main proof is called proof composition, and various works have proposed methods to achieve this~\cite{chase_efficient_2016,agrawal_non-interactive_2018}.
It has been most generically defined in~\cite{campanelliLegoSNARKModularDesign2019}, and is known as the \emph{commit-and-prove} paradigm.
Following this method, proofs can be generated for combined statements, where part of the witness has been committed to beforehand, thereby enabling composition, by linking proofs through witness commitments.
Finally, we note that this approach is also \emph{adaptive}: the commitments may be generated ahead-of-time and independently of the statement.

\textbf{This work.}
It is evident that both collaborative proofs and the commit-and-prove paradigm provide a generic way of dealing with the limitations of most \ac{nizk} schemes.
However, neither approach solves these limitations simultaneously, even though they are of equal importance for improving the efficiency and applicability of \acp{ZKP}.
In this paper, we introduce collaborative \acp{cpnizk}, combining the best of both worlds to enable efficient proving in settings that are otherwise not (natively) supported, and sketch several such use cases for their use (\cref{subsec:example-applications}). Therefore, our work is centered around the following question:

\begin{mdframed}
\begin{center}
  % \textit{How can we construct an efficient \ac{cpsnark} over a distributed witness to enable collaborative proof composition of multiple proving schemes?}
  % \emph{Can users efficiently prove knowledge of a distributed witness using adaptive, composable \ac{nizk} proofs?}
  % \emph{Can proof composition when applied to the distributed witness setting improve performance leading to faster proof generation?}
  \emph{Can users efficiently prove knowledge of a distributed witness using adaptive, composable \ac{nizk} proofs, and does proof composition in this setting result in faster proof generation?}
\end{center}
\end{mdframed}

\textbf{Our contributions.}
We answer this question positively, by defining, implementing and evaluating the novel concept we call collaborative \ac{cpnizk}.
In summary, our contributions are as follows:

\begin{itemize} [leftmargin=*]
    \item We formally define collaborative \acp{cpnizk} (\cref{sec:col-cp}).
    \item We propose two methods for generating Pedersen-like commitments collaboratively (\cref{subsec:collab-comm}).
    \item We design two collaborative \acp{cpnizk} by transforming two existing (commit-and-prove) provers into their \acs{MPC} counterparts (\cref{sec:appr}). Specifically, we construct a collaborative \acs{cpsnark} from LegoGro16~\cite{campanelliLegoSNARKModularDesign2019} (efficient verification and small proof size) and collaborative, commit-and-prove Bulletproofs~\cite{bunz_bulletproofs_2018} (no trusted setup, logarithmic proof size, efficient range proofs).
    \item We implement all our collaborative \acp{cpnizk} and collaborative commitment schemes in Arkworks~\cite{noauthor_textttarkworks_2022} in a modular fashion and make our code available (\cref{sec:impl}). Modularity allows for future extensions based on novel \ac{nizk} and \acs{MPC} schemes.
    \item We evaluate our constructions and demonstrate that our approach incurs only a minor overhead, which is more than compensated for by the benefits of modularity and composition (\cref{sec:eval}). In our experiments, we show several scenarios in which composition significantly improves performance. Splitting large statements into smaller substatements and using a combination of \acp{cpnizk} lead to improved performance compared to using a single \acp{nizk} (\cref{subsec:eval-compare}). Additionally, splitting prover groups improves efficiency by $1.3$--$2\times$. When compared to the construction of~\cite{ozdemirExperimentingCollaborativeZkSNARKs2022}, our protocols show an improvement of $18$--$55\times$ when both constructions are applied in a realistic application scenario (\cref{sec:eval-app}).
\end{itemize}

\subsection{Example Use Cases.}\label{subsec:example-applications}
In this section, we discuss some example use cases, that benefit from collaborative \acp{cpnizk}.
As mentioned in~\cite{ozdemirExperimentingCollaborativeZkSNARKs2022} many real-world and conceptual settings ask for collaborative proofs.
Our collaborative \acp{cpnizk} apply to similar use cases, but additionally offer significant performance improvements and modularity.
Moreover, many applications require users to commit to their data, which makes our construction a more efficient and natural choice.

\textbf{Auditing.}
Many government organizations, regulatory bodies, and financial institutions regularly perform audits on companies or individuals.
Oftentimes, these audits also involve data from external sources, next to information already held by the party under audit.

Consider the process of a mortgage application.
In general, the bank needs to verify that the applicant satisfies certain requirements, i.e., sufficient, stable income, decent credit score, no other mortgages, potential criminal history, et cetera.
The majority of this information is privacy-sensitive.
Since the banks need only verify that the requirements are satisfied, we can significantly reduce the privacy loss by replacing these checks by \acp{ZKP}~\cite{marrikaINGLaunchesZeroKnowledge2017,jurekWhatZeroknowledgeProof2023}.

The applicant can thus construct a proof for all statements together with the external parties that hold the data using collaborative \acp{nizk}. Note that the applicant often cannot access this external data themselves.
Since the number of external parties is rather large, it would be infeasible to create one large collaborative proof.
Fortunately, most statements will be largely disjunct.
Therefore, we can use the composability of collaborative \acp{cpnizk}, by splitting the proof into several parts, such that only a small subset of parties is needed to prove each part.
By using the same commitments for all parts, correctness is guaranteed.

Moreover, the eventual \acp{cpnizk} proof guarantees public verifiability, due to its non-interactive nature.
This allows regulatory bodies to confirm that banks adhere to the rules when giving out mortgages, without seeing any sensitive data. 

The idea, of splitting a large prover group into several smaller parts, can be applied to many practical use cases similar to this. Not only does it improve efficiency between parties, it also reduces the risk of aborts due to adversarial behavior, and removes the need for secure communication channels between all parties.

% \textbf{Statistics from anonymous credentials.}
% Collaborative \acp{cpnizk} also enable combining the use of anonymous credentials with computing statistics in a privacy-preserving way. For instance, consider the case where a group of parties (who do not necessarily trust one another) wish to generate statistics over their combined data. Each participant must also prove ownership of a credential attesting that this data was collected or generated in accordance with some policies or requirements. In this setting, the composability of \acp{cpnizk} allows each party to generate a proof (individually) that the credential meets the policy and link this proof to the collaborative proof generated from computing the statistics. Here, we realistically assume that the anonymous signature scheme includes a commitment to the same input used to generate the statistics and that both schemes support compatible commitment schemes.

\textbf{PA-MPC for e-voting and online auctions.}
\Ac{MPC} allows a group of parties to collaboratively evaluate a function, whilst guaranteeing input privacy and correctness.
Generally, however, correctness cannot be verified by parties that are not involved in the computation.
\Ac{PA-MPC}~\cite{baum_publicly_2014,kanjalkar_publicly_2021} aims to solve this problem by making results publicly verifiable.

While many applications are thinkable, two prominent ones are found in online auctions and e-voting.
In fact, verifiable auctions and e-voting have been an active topic of research for a few decades~\cite{cohenRobustVerifiableCryptographically1985,kikuchi1999multi,schoenmakersSimplePubliclyVerifiable1999,chaumSecretballotReceiptsTrue2004,grothNoninteractiveZeroKnowledgeArguments2005,lee_saver_2019,ramchenUniversallyVerifiableMPC2019,alvarezComprehensiveSurveyPrivacypreserving2020,panjaSecureEndtoendVerifiable2021}.
Both voting and auctions often have a large number of participants, making it computationally infeasible to involve all parties in the computation.
Rather, the computation is run by a fixed group of servers, that take inputs from participants and compute the result, i.e., highest bid or vote tally.
Generally, participants do not send their plain input to the server but rather a transformation that hides its true value, e.g., a secret sharing, hiding commitment, or blinded value, depending on the privacy-preserving protocol that is used.

So, while participants can guarantee privacy of their own inputs, they have no control over the results being computed correctly.
But, by using \ac{PA-MPC} and adding a correctness proof to the computation, each participant can verify the correctness of the result, thereby guaranteeing an honest auction or voting process.
To verify the proof, one often needs access to all commitments to the original input data.
Collaborative \acp{cpnizk} are a very efficient way of constructing proofs for committed data, as we argue in~\cref{subsec:pa-mpc}.

\textbf{Improving efficiency through modularity.}
Collaborative \acp{cpnizk} can be constructed from any combination of \ac{MPC} and \ac{cpnizk} schemes, as we show in~\cref{sec:col-cp} and onwards.
Moreover, we note that many regular \acp{nizk} can be transformed into \acp{cpnizk} for select commitment schemes~\cite{campanelliLegoSNARKModularDesign2019}.

The modularity of our construction in combination with its composition property, i.e., we can compose proofs for differing statements on partially overlapping commitments, gives its users significant freedom.
This freedom can be used to improve efficiency, by cleverly combining \ac{MPC} and \ac{cpnizk} schemes, depending on the users' requirements and proof statements.
This potential for efficiency improvements is especially useful to reduce the overhead introduced by \ac{MPC}, as exemplified in the use cases above.

\section{Related Work}\label{sec:related-work}
Below, we give an overview of prior works on (adaptive) composability and distributed \ac{ZKP} provers and discuss their relation to our unifying construction.
In~\cref{apx-additional-relatedwork}, we also provide an additional discussion on existing approaches for distributed provers on committed data in \ac{PA-MPC}, which unlike our work are tailored to one specific, often monolithic, combination of \ac{MPC}, commitment, and \acp{nizk} scheme.
% Contrastingly, the generality of our scheme allows for modularity and composability of different schemes, leading to the earlier mentioned advantages.

\textbf{Commit-and-Prove constructions.}
The concept of combining different \acp{nizk} to improve efficiency when handling heterogeneous statements, i.e., handling substatements of a larger proof statement with tailored \ac{nizk} schemes has previously been considered in~\cite{chase_efficient_2016, agrawal_non-interactive_2018}.
These works, however, only focus on combining select schemes, e.g.,~\cite{agrawal_non-interactive_2018} combines Pinocchio~\cite{parno_pinocchio_2013} with $\Sigma$-protocols.

Our work builds upon the definition for \acp{cpsnark} as given in LegoSNARK~\cite{campanelliLegoSNARKModularDesign2019}.
In their work, the authors flexibly define \acp{cpsnark}, i.e., succinct \acp{cpnizk}.
And, they show how to compose different \acp{cpsnark} for Pedersen-like and polynomial commitments.
Additionally, they define the concept of \acp{ccsnark}, a weaker form of \acp{cpsnark}, and show how to compile \acp{ccsnark} into \acp{cpsnark}.
Specifically, they introduce a commit-and-prove version of Groth16~\cite{groth_size_2016} for Pedersen-like commitments, that we make collaborative in this work.

Next to this, we observe a number of schemes with commit-and-prove functionality that were introduced before LegoSNARK\@.
Geppetto~\cite{costello_geppetto_2015} can be used to create proofs on committed data, however, the commitment keys depend on the relation to be proven, i.e., it is not adaptive.
Adaptive Pinocchio~\cite{veeningen_pinocchio-based_2017} removes this restriction, making it \iac{cpsnark} according to the definition of LegoSNARK\@.
Other \acp{cpsnark} are presented in~\cite{lipmaa_prover-efficient_2017} and~\cite{groth_short_2010}, however, they rely on specific commitment schemes, that may not be composable with other \ac{nizk} schemes.
Finally, we note that~\cite{costello_geppetto_2015} and~\cite{lipmaa_prover-efficient_2017} give alternative definitions for \acp{cpsnark}, however, we choose LegoSNARK over these definitions due to its flexibility.

\textbf{Distributed provers.}
The notion of distributed provers, i.e., $N$ parties respectively knowing $N$ witnesses, has already been considered by Pedersen in 1991~\cite{pedersenDistributedProversApplications1991}. He describes the general paradigm of distributing prover algorithms using \ac{MPC} and presents applications to undeniable signatures.
More recently, the authors of Bulletproofs~\cite{bunz_bulletproofs_2018}, discuss how to use \ac{MPC} to create aggregated range proofs, however this does not extend to generic arithmetic circuits.
Both solutions are specific to one honest-majority \ac{MPC} scheme and not very efficient.

The authors of~\cite{dayamaHowProveAny2021} aim to improve upon existing work, by introducing a compiler to generate distributed provers for \acp{ZKP} based on the \ac{IOP} paradigm~\cite{ben-sassonInteractiveOracleProofs2016}.
They construct distributed versions of Ligero~\cite{amesLigeroLightweightSublinear2017}, Aurora~\cite{ben-sassonAuroraTransparentSuccinct2019}, and a novel scheme called Graphene.
Unlike our work, their solution is restricted to a particular subset of \acp{ZKP}, and does not consider composability.

Another solution, that we use as inspiration for our constructions is presented by Ozdemir and Boneh\cite{ozdemirExperimentingCollaborativeZkSNARKs2022}, who define collaborative \acp{zksnark}.
They show how to construct these from generic \acp{zksnark} in combination with \ac{MPC}.
Their implementation is evaluated for Groth16~\cite{groth_size_2016}, Marlin~\cite{chiesa_marlin_2019}, and Plonk~\cite{gabizon_plonk_2019}.
Finally, they describe how to implement collaborative Fractal~\cite{chiesaFractalPostquantumTransparent2020}.

Additionally, we note three schemes that consider distributing the prover to increase efficiency, e.g., for outsourcing purposes: DIZK~\cite{wu_dizk_2018}, Eos~\cite{chiesa_eos_2023}, and zkSaaS~\cite{garg_zksaas_2023}.
DIZK distributes the prover across multiple machines in one cluster, by identifying the computationally costly tasks and distributing these. Yet, their method is not privacy-preserving and thus not applicable to our setting. zkSaaS~\cite{garg_zksaas_2023}, on the other hand, does guarantee privacy-preservation when outsourcing the prover to a group of untrusted servers.
The authors evaluate their construction for three \acp{zksnark}: (Groth16~\cite{groth_size_2016}, Marlin~\cite{chiesa_marlin_2019}, and~\cite{gabizon_plonk_2019}). Eos achieves similar results, but only considers \acp{zksnark} based on polynomial \acp{IOP} and polynomial commitments, such as Marlin~\cite{chiesa_marlin_2019}.
Neither Eos nor zkSaaS considers composability and modularity, and thus their advantages.

\textbf{Other notions of distribution.}
Lastly, we observe that distribution is also considered for other \ac{ZKP} aspects.
Feta~\cite{baumFetaEfficientThreshold2022} considers threshold distributed verification in the designated-verifier setting.
To increase the reliability of the trusted setup, several works present solutions for replacing the setup algorithm by a distributed setup ceremony~\cite{ben-sassonSecureSamplingPublic2015,kohlweissSnarkyCeremonies2021}.

\section{Preliminaries} \label{sec:prelim}
The finite field of integers modulo a prime $p$ is denoted by $\ZZ_p$. Additionally, \GG is a cyclic group of order $q$, where $g, h \in \GG$ are arbitrary generators, all group operations are written multiplicatively.
$x \overset{\$}{\leftarrow} \ZZ_p^*$ denotes uniform random sampling of an element from $\ZZ_p^*$.

Some \ac{nizk} schemes make heavy use of bilinear groups $(q, \GG_1, \GG_2, \GG_T, e)$, where $\GG_1$, $\GG_2$, and $\GG_T$ are of prime order $q$.
$e: \GG_1 \times \GG_2 \rightarrow \GG_T$ is a bilinear map, or \emph{pairing}, such that $e(h_1^a, h_2^b) = e(h_1, h_2)^{ab}$, for all $a, b \in \FF_q$ and all $h_1 \in \GG_1$, $h_2 \in \GG_2$.

\emph{Vectors and matrices.}
We denote vectors as $\mathbf{v}$ and matrices as $\mathbf{A}$. For example, $\mathbf{A} \in \ZZ^{n \times m}$ represents a matrix with $n$ rows and $m$ columns, where $a_{i,j}$ is the element located in the $i$th row and $j$th column of $\mathbf{A}$. We denote multiplication of a scalar $c \in \ZZ_p$ and vector $\mathbf{a} \in \ZZ_p^n$ as $\mathbf{b} = c \cdot \mathbf{a} \in \ZZ_p^n$, with $b_i = c \cdot a_i$.

Additionally, let $\mathbf{g} = (g_1, \ldots, g_n) \in \mathbb{G}^n$ and $\mathbf{x} = (x_1, \ldots, x_n) \in \ZZ_p^n$, then $\mathbf{g}^\mathbf{x}$ is defined as $\mathbf{g}^\mathbf{x} = g_1^{x_1} \cdots g_n^{x_n}$. Also, $\langle \mathbf{a}, \mathbf{b} \rangle = \sum_{i=1}^n a_i \cdot b_i$ denotes the inner product of two vectors $\mathbf{a}, \mathbf{b} \in \ZZ^n$, and $\mathbf{a} \circ \mathbf{b} = (a_1 \cdot b_1, \ldots, a_n \cdot b_n)$ their Hadamard product.

A vector polynomial is denoted by $p(X) = \sum_{i=0}^d \mathbf{p_i} \cdot X^i \in \ZZ_p^n[X]$ where coefficients are $n$-vectors, i.e., $\mathbf{p_i}\in\ZZ_p^n$. Given $0 \leq \ell \leq n$, we use $\mathbf{a}[\slice\ell] = (a_1, \ldots, a_\ell) \in \ZZ^\ell$ to denote the left slice and $\mathbf{a}[\ell\slice] = (a_{\ell+1}, \ldots, a_n) \in \ZZ^{n-\ell}$ for the right slice.

\subsection{Multiparty computation} \label{subsec:mpc}
\Ac{MPC} is a collection of techniques that allow multiple parties to jointly compute a function on their combined inputs while preserving privacy of their respective inputs.
Typically, $N$ parties $P_1, \ldots, P_N$ collaboratively evaluate $f \colon \mathcal{X}^N \rightarrow Y$, where the $i$-th party $P_i$ holds input $x_i \in \mathcal{X}$.
A secure \ac{MPC} protocol \protocol should at least satisfy \emph{privacy}, \emph{correctness}, and \emph{independence of inputs}~\cite{lindellSecureMultipartyComputation2020}.
Informally, this means that no party should learn anything other than its prescribed output, that each honest party is guaranteed correctness of its received output, and that corrupted parties choose their inputs independently of honest parties' inputs.

Concretely, we say that $\Pi$ is secure against $t$ corrupted parties, if $\Pi$ is secure against any adversary \adv corrupting no more than $t$ out of $N$ parties.
Furthermore, security depends on the type of behaviour that corrupted parties may show.
For \emph{semi-honest}, or \emph{honest-but-curious}, adversaries, corrupted parties follow the protocol specification as prescribed.
\adv simply tries to learn as much information as possible from the corrupted parties' states.
Adversely, in the \emph{malicious} model, corrupted parties may deviate \emph{arbitrarily} from the specified protocol to try and break security.

The majority of practical \ac{MPC} protocols satisfy security in either model, only when \emph{aborts} are allowed.
In this \emph{security-with-abort} model, the adversary may cause an abort of the protocol, thereby denying outputs to honest parties.
We refer the reader to~\cite{Goldreich2004_foundations} for all formal definitions related to \ac{MPC}.

\textbf{Arithmetic circuits for MPC.} An \emph{arithmetic circuit} allows us to describe a large class of computations using elementary gates.
Many generic \ac{MPC} frameworks can be used to securely evaluate bounded size arithmetic circuits.
A circuit is represented by \iac{DAG} consisting of wires and gates.
Wires carry values in a finite field $\ZZ_p$, and can be either public or private.
Each gate takes two input values and returns one output, either a \emph{multiplication} or \emph{addition} of its inputs.
A subset of the wires is dedicated for assigning the circuit inputs and outputs.

\textbf{SPDZ.}
Many \ac{MPC} protocols are based on \emph{secret sharing}, in which private inputs and intermediate values are represented by so-called \emph{shares}.
Oftentimes, inputs are secret shared using a $t$-out-of-$N$ secret sharing scheme, where each value is split in $N$ shares (one share per party), such that given at least $t$ shares one can reconstruct the full input, and otherwise no information is revealed.

SPDZ~\cite{damgard_multiparty_2011} is \iac{MPC} framework based on \emph{additive $N$-out-of-$N$ secret sharing}, i.e., a value $y \in \ZZ_p$ is split into $y_1, \ldots y_N$, such that all shares together sum to $y$.
We will use $\llbracket y \rrbracket_i$ to denote the share of $y$ held by party $i$. Additionally, $\llbracket y \rrbracket$ will denote the vector of all shares of $y$.
The SPDZ framework can be used to evaluate arbitrary arithmetic circuits of bounded size in the presence of \emph{malicious} adversaries.

All shares in SPDZ are associated with \iac{MAC} that is used to maintain correctness of the computation.
% The MAC key is secret shared among the parties and is used to verify the integrity of the shares at the end of the protocol.
A share together with its \ac{MAC} is called an authenticated share.
SPDZ uses authenticated shares to achieve active security, meaning that when a corrupted party deviates from the protocol, the honest parties can detect this and abort the protocol.
SPDZ is thus a secure-with-abort protocol, being able to handle up to and including $N-1$ corrupted parties.
Since, each operation on an authenticated share updates both the value and its \ac{MAC}, each operation takes double the time of a regular operation.

Note that, due to its additive property, arithmetic operations between public and shared values, as well as addition of shared values can be done `for free', i.e., locally.
However, multiplying two shared values requires some form of interaction.
To solve this problem efficiently, SPDZ adopts somewhat homomorphic encryption to create a sufficiently large amount of Beaver triplets~\cite{beaver_efficient_1992} in a circuit-independent \emph{offline} phase.
These triplets can be used to multiply shared values more efficiently during the \emph{online} phase.

\subsection{Commitments}\label{subsec:commitments}
A non-interactive commitment scheme is a 3-tuple of \ac{ppt} algorithms:
\begin{itemize}[leftmargin=*]
    \item $\setup(\secparam) \rightarrow \ck$: Given the security parameter \secpar, this outputs the commitment key \ck, which also describes the input space \domain, commitment space $\mathcal{C}$, and opening space $\mathcal{O}$;
    \item $\commit(\ck, u, o) \rightarrow c$: Given \ck, an input $u$ and opening $o$, this outputs a commitment $c$;
    \item $\vercommit(\ck, c, u,o) \rightarrow \bin$: Asserts whether $(u,o)$ opens the commitment $c$ under \ck, returns 1 (accept) or 0 (reject),
\end{itemize}
and should satisfy at least the following properties:
\begin{itemize}[leftmargin=*]
    \item \emph{Correctness:} For all \ck and for each $u \in \domain$, $o \in \mathcal{O}$, if $c = \commit(\ck,u,o)$, then $\vercommit(\ck, c, u,o) = 1$;
    \item \emph{Binding:} Given a commitment $c$ to an input-opening pair $(u,o)$, it should be hard to find $(u', o')$ with $u' \neq u$, such that $\vercommit(\ck, c, u', o')=1$;
    \item \emph{Hiding:} For any \ck and every pair $u, u' \in \domain$, the distribution (with respect to $o$) of $\commit(\ck, u, o)$ should be indistinguishable from that of $\commit(\ck, u', o)$.
\end{itemize}

\subsection{Zero Knowledge Proofs} \label{subsec:zkp}
\Acp{ZKP}~\cite{goldwasserKnowledgeComplexityInteractive1989} allow a prover \pdv to convince a verifier \vdv of the existence of a secret \emph{witness} $w$ for a public \emph{statement} $x$, such that both satisfy an \npol-relation $\relation \in \relationfamily$, i.e., $(x,w) \in \relation$, for some family of relations \relationfamily.
In this work, we focus on \ac{nizk} proofs of knowledge, in which \pdv not only proves existence of $w$, but also that it \emph{knows} $w$.
Specifically, we consider NIZKs in the \ac{CRS} model~\cite{blum_non-interactive_1988,desantisRobustNoninteractiveZero2001,groth_size_2016} (where the \ac{CRS} can be reused for any polynomial number of proofs), which are defined by a 3-tuple of \ac{ppt} algorithms:
\begin{itemize}[leftmargin=*]
    \item $\kgen(\secparam, R) \rightarrow (\ek, \vk)$: takes the security parameter \secpar and a relation \relation and outputs the \ac{CRS} $(\ek, \vk)$.
    \item $\prove(\ek,x,w)$. Given the evaluation key \ek, statement $x$, and witness $w$, generates a valid proof $\pi$ (if $(x,w) \in \relation$).
    \item $\verify(\vk,x,\pi) \rightarrow \bin$. Given the verification key \vk, $x$, and $\pi$, outputs 1 if $\pi$ is valid and 0 otherwise.
\end{itemize}

A secure \ac{nizk} proof of knowledge should at least satisfy the following (informal) properties:
\begin{itemize}[leftmargin=*]
    \item \emph{Completeness:} Given a true statement $(x,w) \in \relation$, an honest prover \pdv is able to convince an honest verifier \vdv.
    \item \emph{Knowledge soundness:} For any prover \pdv, there exists an extractor $\edv^\pdv$ that produces a witness $w$ such that $(x,w) \in \relation$, whenever \pdv convinces \vdv for any given $x$.
    \item \emph{Zero-knowledge:} A proof $\pi$ should reveal no information other than the truth of the statement $x$.
\end{itemize}
We refer to~\cite{desantisRobustNoninteractiveZero2001} for their formal definitions.
When knowledge soundness only holds computationally, the scheme is called an \emph{argument} of knowledge, rather than a \emph{proof}.
For simplicity, we refer to both as proofs in the remainder of this work unless we wish to specify it explicitly.
An example of a secure \iac{nizk}, that is also used in our work, is \emph{Bulletproofs}~\cite{bunz_bulletproofs_2018}, where non-interactiveness follows from the Fiat-Shamir heuristic~\cite{fiat_how_1987}.

When \iac{nizk} argument of knowledge is also succinct, i.e., the verifier runs in $\poly[\secpar + |x|]$ time and the proof size is $\poly$, we call it \iac{zksnark}~\cite{bitansky_extractable_2012}. In this work, we build upon the \emph{Groth16}~\cite{groth_size_2016} \ac{zksnark}, to implement a collaborative \ac{cpsnark}.

\subsection{Collaborative NIZKs} \label{subsec:col-nizks}
Ozdemir and Boneh~\cite{ozdemirExperimentingCollaborativeZkSNARKs2022} define collaborative \acp{zksnark}, building upon previous definitions~\cite{ben-sassonInteractiveOracleProofs2016}.
They consider the setting of multiple provers wanting to prove a statement about a distributed witness, i.e., each party holds a vector of secret shares that together reconstruct the full witness.

When leaving out the requirement of \emph{succinctness}, we obtain the general definition for collaborative \acp{nizk} arguments of knowledge.
These are defined analogously to regular \acp{nizk} by a 3-tuple $(\kgen, \protocol, \verify)$, where \kgen and \verify have the same definition.
\protocol is an interactive protocol serving to replace \prove, where $N$ provers with private inputs $w_1, \ldots, w_N$ together create a proof $\pi$ for a given statement $x$.

\emph{Completeness} and \emph{knowledge soundness} for collaborative \acp{nizk} arguments are defined analogously to their non-collaborative counterpart.
On the other hand, \emph{zero-knowledge} is replaced by the notion of \emph{$t$-zero-knowledge}, which guarantees that $t$ colluding, malicious provers cannot learn anything about the witnesses of other provers, other than the validity of the full witness.
The relation between regular and collaborative \acp{nizk} follows from~\cite{ozdemirExperimentingCollaborativeZkSNARKs2022}:
\begin{theorem}\label{thm:collab-nizk}
If $(\kgen,\prove,\verify)$ is \iac{nizk} argument of knowledge for \relationfamily, and \protocol \iac{MPC} protocol for \prove for $N$ parties that is secure-with-abort against $t$ corruptions, then $(\kgen,\protocol,\verify)$ is a collaborative \ac{nizk} argument of knowledge for \relationfamily. \emph{(Proof follows directly from~\cite[Thm. 1]{ozdemirExperimentingCollaborativeZkSNARKs2022}.)}
\end{theorem}

\section{Collaborative CP-NIZKs} \label{sec:col-cp}
A collaborative \ac{cpnizk} argument of knowledge is an argument that, given $x$, is used to prove knowledge of $\mathbf{w} = (w_1, \ldots, w_N)$ such that $(x,w_1,\ldots,w_N) \in \relation$, where $w_i=(u_i, \omega_i)$ and $\mathbf{u} = (u_1, \ldots, u_N)$ opens a commitment $c_{\mathbf{u}}$.
It is collaborative in the sense that $N$ provers, who have \emph{distributed knowledge}~\cite{halpernKnowledgeCommonKnowledge1990} of the witness vector $\mathbf{w}$, together construct a single argument of knowledge.

In practice, the commitment to $\mathbf{u}$ might be split over $\ell$ commitments $\set{c_j}_{j \in [\ell]}$.
As noted in~\cite{campanelliLegoSNARKModularDesign2019}, this splitting is crucial to efficiently exploit the compositional power of \acp{cpnizk}, which we explain shortly after.
We assume that the specification of this splitting is described in \relation.

Formally, a \acs*{cpnizk} argument of knowledge is defined by a 3-tuple of \ac{ppt} algorithms:
\begin{itemize}[leftmargin=*]
    \item $\kgen(\secparam,\ck,\relation) \rightarrow (\ek, \vk)$: Given the security parameter \secpar, commitment key \ck, and relation \relation, this algorithm outputs the \ac{CRS}, i.e., the evaluation \ek and verification \vk keys;
    \item $\protocol(\ek, x, \set{c_j}_{j \in [\ell]}, \set{\mathbf{u}_j}_{j \in [\ell]}, \set{\mathbf{o}_j}_{j \in [\ell]}, \boldsymbol{\omega}) \rightarrow \pi$: In this interactive protocol, given $\ek$, statement $x$, and commitments $c_j$, all $N$ provers use their respective private inputs, i.e., commitment openings $(u_i, o_i)_j$ and non-committed witness $\omega_i$, to produce a proof $\pi$;
    \item $\verify(\vk, x, \set{c_j}_{j \in [\ell]}, \pi) \rightarrow \bin$: Given \vk, $x$, all $c_j$'s, and $\pi$, this algorithm returns 1 when $\pi$ is a valid proof and 0 otherwise.
\end{itemize}

Building upon definitions for \acp{cpsnark} in~\cite{campanelliLegoSNARKModularDesign2019}, we define commitment-enhanced (CE) relations as follows.
\begin{definition}[Commitment-enhanced relation]\label{def:ce-relation}
Let \relationfamily be a family of relations \relation over $\domain_x \times \domain_{\mathbf{u}} \times \domain_{\boldsymbol{\omega}}$, such that $\domain_{\mathbf{u}} = \domain_{\mathbf{u}_1} \times \cdots \times \domain_{\mathbf{u}_\ell}$. Given a commitment scheme $\comm=(\setup,\commit,\vercommit)$ with input space $\domain$, commitment space $\mathcal{C}$ and opening space $\mathcal{O}$, such that $\domain_{\mathbf{u}_j} \subseteq \domain, \forall j \in [\ell]$.
The CE version of \relationfamily is given by $\relationfamily^\comm$, such that:
\begin{itemize}[leftmargin=*]
    \item every $\relation^\text{CE} \in \relationfamily^\comm$ can be represented by a pair $(\ck, \relation)$ where \ck is a possible output from $\setup(\secparam)$, $\relation \in \relationfamily$; and
    \item $\relation^\text{CE}$ is the collection of statement-witness pairs $(\phi; \mathbf{w})$, with $\phi = (x, \set{c_j}) \in \domain_x \times \domain^\ell$ and $\mathbf{w} = (\set{\mathbf{u}_j}, \set{\mathbf{o}}, \boldsymbol{\omega}) \in \domain^{\ell} \times \mathcal{O}^{\ell} \times \domain_{\boldsymbol{\omega}}$ such that $\vercommit(\ck, c_j, \mathbf{u}_j, \mathbf{o}_j) = 1, \forall j \in [\ell]$ and $(x, \set{\mathbf{u}_j}, \boldsymbol{\omega}) \in \relation$.
\end{itemize}
\end{definition}

\subsection{Security properties}

Using \cref{def:ce-relation} and building upon~\cite{ozdemirExperimentingCollaborativeZkSNARKs2022,ben-sassonInteractiveOracleProofs2016}, we present a unified definition for collaborative \acp{cpnizk} arguments:

\begin{definition}[Collaborative \acs*{cpnizk} argument of knowledge]\label{def:collab-cp-nizk}
Given a family \relationfamily of relations \relation over $\domain_x \times \domain_{\mathbf{u}} \times \domain_{\boldsymbol{\omega}}$ and a commitment scheme $\comm$.
Let $\mathcal{U}(\secpar)$ denote the uniform distribution over all functions $\rho: \bin^* \rightarrow \bin^\secpar$, i.e., if $\rho \sample \mathcal{U}(\secpar)$, then $\rho$ is a \emph{random oracle}.
In the random oracle model, $(\kgen, \protocol, \verify)$ is a collaborative \acs*{cpnizk} argument of knowledge for $\relationfamily$ and \comm with $N$ provers, secure against $t$ malicious provers if for all $\secpar \in \NN$, $R \in \relationfamily$, with corresponding $R^\text{CE} \in \relationfamily^\comm$ as in \cref{def:ce-relation}:
\begin{enumerate}[leftmargin=*]
    \item \emph{Completeness}: 
    For all $(x,\set{c_j}; \set{\mathbf{u}_j}, \set{\mathbf{o}_j}, \boldsymbol{\omega}) \in \relation^\text{CE}$ and \ck, the following is negligible in \secpar:
    \begin{multline*}
        \Pr\left[
        \verify^\rho(\vk,x,\set{c_j},\pi) = 0
        \middle\vert \right. \\
        \left.
        \begin{array}{l}
             \rho \sample \mathcal{U}(\secpar)  \\
             (\ek,\vk) \leftarrow \kgen^\rho(\secparam,\ck,\relation) \\
             \pi \leftarrow \Pi^\rho(\ek, x, \set{c_j}, \set{\mathbf{u}_j}, \set{\mathbf{o}_j}, \boldsymbol{\omega})
        \end{array}
        \right]
    \end{multline*}
    
    \item \emph{Knowledge soundness}: For all $(x,\set{c_j})$, \ck and \ac{ppt} provers $\vec{\pdv}=(\pdv_1^*, \ldots, \pdv_N^*)$, there exists \iac{ppt} extractor \edv such that the following is negligible in \secpar:
    \begin{multline*}
        \Pr\left[
        \begin{array}{l}
            (x, \set{c_j}; \mathbf{w}) \not\in \relation^\text{CE} \\
            \verify^\rho(\vk,x,\set{c_j},\pi) = 1
        \end{array}
        \middle\vert \right. \\
        \left.
        \begin{array}{l}
        \rho \sample \mathcal{U}(\secpar) \\
        (\ek,\vk) \leftarrow \kgen^\rho(\secparam,\ck,\relation) \\
         \pi \leftarrow \vec{\pdv}^\rho(\ek, \vk, x, \set{c_j}) \\
         \mathbf{w} \leftarrow \edv^{\rho, \vec{\pdv}^\rho}(\ek, \vk, x, \set{c_j})
        \end{array}
        \right]
    \end{multline*}
    Where $\edv^{\rho, \vec{\pdv}^\rho}$ means that \edv has access to the random oracle $\rho$ and can re-run $\vec{\pdv}^\rho(\ek, x, \set{c_j})$, where \edv can reprogram $\rho$ each time, and only receives the output of $\vec{\pdv}$.
    
    \item \emph{t-zero-knowledge}: For all \ac{ppt} adversaries \adv controlling $k \leq t$ provers $\pdv_{i_1}, \ldots, \pdv_{i_k}$, there exists \iac{ppt} simulator \sdv such that for all $(x,c; \mathbf{w})$ and \ac{ppt} distinguishers $\ddv$
    \begin{multline*}
        \condprob{\ddv^\rho(\transcript)=1}
        {\begin{array}{l}
        \rho \sample \mathcal{U}(\secpar) \\
        (\ek,\vk) \rightarrow \kgen^\rho(\secparam, \relation) \\
        b \leftarrow (x, \set{c_j}; \mathbf{w}) \in \relation^\text{CE} \\
        (\transcript, \mu) \leftarrow \sdv^\rho(\ek,\vk,x,\set{c_j},\\
        ~~~~~~~~~~~~~~~~~w_{i_1},\ldots,w_{i_k}, b)
        \end{array}} \overset{c}{\equiv} \\
        \condprob{\ddv^{\rho[\mu]}(\transcript)=1}
        {\begin{array}{l}
        \rho \sample \mathcal{U}(\secpar) \\
        (\ek,\vk) \rightarrow \kgen^\rho(\secparam, \relation) \\
        \transcript \leftarrow \view_\adv^\rho(\ek,\vk,x,\set{c_j},\mathbf{w})
        \end{array}},
        \end{multline*}
    where \transcript denotes a transcript. $\view_\adv^\rho$ is the view of \adv when the provers interact for the given inputs, where the honest provers follow \protocol and dishonest provers may deviate. $\mu$ is a partial function, and $\rho[\mu]$ is the function that returns $\mu(x)$ if $\mu$ is defined on $x$ and $\rho(x)$ otherwise.
\end{enumerate}
\end{definition}

As a corollary, we obtain the following definition:

\begin{definition}[cCP-SNARK]
A collaborative \acs*{cpsnark} is a collaborative \acs*{cpnizk} that additionally satisfies \emph{succinctness}: the verifier runs in $\poly[\secpar + |x|]$ time and the proof size is $\poly$.
\end{definition}

\emph{Completeness}, \emph{soundness}, and \emph{zero-knowledge} are defined similarly to their (non-collaborative) \ac{cpnizk} counterparts, however there are some distinctions (in agreement with~\cite{ozdemirExperimentingCollaborativeZkSNARKs2022}).

For collaborative proofs, knowledge soundness only guarantees that the $N$ provers together `know' the complete witness vector $\mathbf{w}$. However, it does not state anything about how this knowledge is distributed over the provers, i.e., there is no guarantee that party $\pdv_i$ knows $w_i$.
Actually, an honest verifier, observing only the proof $\pi$, is not even able to determine the number of provers $N$ that participated in \protocol.
If desired, this `limitation' could be circumvented, by including additional conditions in \relation that link elements of $\mathbf{w}$ to secrets that are known to be in possession of $\pdv_i$.

In the regular definition of zero-knowledge, the simulator \sdv should be able to simulate a proof $\pi$ for any $(x,\set{c_j},\mathbf{w}) \in \relation^\text{CE}$, without having access to $w$.
Observe that, in this case, it is sufficient to consider only valid pairs $(x,w)$, since we consider an honest prover.
However, for \emph{$t$-zero-knowledge}, up to $t$ provers may act maliciously, implying that we cannot make the same assumption.
Actually, each malicious prover $\pdv_{i_k}$ can arbitrarily choose $w_{i_k}$.
This is modelled by providing all $w_{i_k}$'s to $\sdv$ along with a bit $b$, denoting whether $(x, \set{c_j},\mathbf{w}) \in \relation^\text{CE}$.
All in all, $t$-zero-knowledgeness guarantees that, for no more than $t$ malicious provers, nothing other than the validity of the witness $\mathbf{w}$ and (possibly) the entries $w_{i_k}$ of the malicious provers is revealed.

\begin{theorem}
    Let $(\kgen, \prove, \verify)$ be \iac{cpnizk} argument of knowledge for \relationfamily and \comm, and \protocol \iac{MPC} protocol for \prove for $N$ parties that has security-with-abort against $t$ corrupted parties, then $(\kgen, \protocol, \verify)$ is a collaborative \ac{cpnizk} argument of knowledge for \relationfamily and \comm with $N$ provers, secure against $t$ malicious provers.
\end{theorem}
\begin{proof}
    \Iac{cpnizk} argument of knowledge for \relationfamily and \comm is a \nizk argument of knowledge for $\relationfamily^\comm$ as defined in \cref{def:ce-relation}, where $\mathbf{u}$, $\boldsymbol{\omega}$, $\mathbf{o}$, $\mathbf{u}_j$ are all single-element vectors, i.e., there is only one prover.
    
    Similarly, a collaborative \iac{cpnizk} argument of knowledge with $N$ provers for \relationfamily and \comm is a collaborative \nizk argument of knowledge for $\relationfamily^\comm$ as defined in \cref{def:ce-relation}, where all vectors are of length $N$.
    
    Thus, we are essentially in the situation of \cref{thm:collab-nizk} and can conclude that security holds.
\end{proof}

We have opted to define security in the \emph{random oracle model} in \cref{def:collab-cp-nizk}, since some \ac{nizk} schemes are only proven secure in the random oracle model, such as those based on the Fiat-Shamir paradigm~\cite{goldwasserSecurityFiatShamirParadigm2003,fiat_how_1987}, e.g., non-interactive \emph{Bulletproofs}~\cite{bunz_bulletproofs_2018}.
However, not all \ac{nizk} schemes depend upon the random oracle model for security, e.g., \emph{Groth16}~\cite{groth_size_2016}, in which case security can be defined analogously, but without giving parties access to a random oracle.

\subsection{Composition of collaborative CP-NIZKs}\label{subsec:composition}
To underscore the relevance of collaborative \acp{cpnizk}, we show how to compose two different \ac{cpnizk} schemes on (partially) overlapping commitments and what benefits come with.
We specifically focus on the novel, added benefits of composability for the collaborative setting.

\textbf{Composition of relations with shared inputs.} Building upon~\cite{campanelliLegoSNARKModularDesign2019}, consider two families of relations $\relationfamily^a$ and $\relationfamily^b$, such that $\domain_{\mathbf{u}}^a$ and $\domain_{\mathbf{u}}^b$ can be split as: $\domain_{\mathbf{u}}^a = \domain_{\mathbf{u}}^0 \times \domain_{\mathbf{u}}^1$ and $\domain_{\mathbf{u}}^b = \domain_{\mathbf{u}}^0 \times \domain_{\mathbf{u}}^2$.
In other words, we consider two relations where part of the committed inputs (those in $\domain_{\mathbf{u}}^0$) are identical.
We define the family of conjunctions of these relations as $\relationfamily^{\land} = \set{\relation^{\land}_{\relation^a, \relation^b} : \relation^a \in \relationfamily^a, \relation^b \in \relationfamily^b}$, where $(x_a, x_b, \mathbf{u}_0, \mathbf{u}_1, \mathbf{u}_2, (\boldsymbol{\omega}_a, \boldsymbol{\omega}_b)) \in \relation^{\land}_{\relation^a, \relation^b}$ if and only if $(x_a, \mathbf{u}_0, \mathbf{u}_1, \boldsymbol{\omega}_a) \in \relation^a \land (x_b, \mathbf{u}_0, \mathbf{u}_2, \boldsymbol{\omega}_b) \in \relation^b$.

Given a commitment scheme \comm, let $\ccp^{a/b}$ be the respective collaborative \ac{cpnizk} arguments of knowledge for $\relationfamily^{a/b}$ and \comm.
Building a collaborative \ac{cpnizk} argument $\ccp^{\land}$ for $\relationfamily^{\land}$ and \comm from this is straightforward.
It is sufficient to use the algorithms provided by $\ccp^{a/b}$ to compute two proofs $\pi_a$ and $\pi_b$ on the same commitment $c^s$ to $\mathbf{u}^s$.
Verifying these proofs for the same value of $c^s$ is sufficient to guarantee correctness of the conjunction.
For a more formal description and security proof of the non-collaborative construction we refer to~\cite{campanelliLegoSNARKModularDesign2019} and note that our construction can be defined and proven secure analogously.

Next to this, we observe that this composition can be readily extended to more than two relations, by applying the above trick multiple times.
Finally, we note that a disjunction of relations can be efficiently constructed using a simple trick with two additional witness entries, following~\cite{campanelliLegoSNARKModularDesign2019}.

\textbf{Composing proofs with distinct prover sets.}
A useful observation, which opens an efficient way of tackling a wide range of practical use cases, is that the prover sets for $\ccp^{a/b}$ need not be equal.
I.e., it is possible for two sets of provers with only partial overlap\footnote{As a matter of fact, the prover sets need not have any overlap at all.} to prove a single relation, as long as this relation can be split into parts.
We can simply adopt the construction as described above for proving compositions of relations and use different prover sets for different `subproofs'.

\textbf{Improving efficiency and security by composition.}
We observe 4 main ways of using the composability of collaborative \acp{cpnizk} to significantly improve practical efficiency and security properties:
\begin{itemize}[leftmargin=*]
    \item By splitting a conjunction of several statements, such that each smaller statement contains witnesses held by a smaller set of provers, we can let each smaller statement be proven only by this smaller set of provers.
    Since many \ac{MPC} protocols have communication that scales quadratically in the number of parties, having multiple smaller prover sets is likely much more efficient. Moreover, many \ac{MPC} protocols require direct communication channels between each party and that all parties are present simultaneously during the online phase. This quickly becomes infeasible for large prover groups. By splitting the statement and prover sets into smaller parts, we avoid this issue.
    \item By splitting a large statement into several smaller ones with different prover sets for each, we can also use different \ac{MPC} schemes for each substatement.
    This is useful in settings where different subsets of provers put different security requirements than others. By splitting the proof along these subsets, one could for example use a less efficient, dishonest majority \ac{MPC} scheme where needed and a more efficient honest majority \ac{MPC} scheme where possible.
    \item Similarly, one can use different \ac{nizk} schemes for different substatements. This can improve efficiency, as some specialized \ac{nizk} schemes are more efficient at proving certain statements than other generic schemes. Moreover, we can improve the security guarantees by using different \ac{nizk} schemes. For example, one could avoid schemes with trusted setups for critical substatements, and only use those with a trusted setup
    % (generally more efficient) 
    for less critical statements.
    \item Finally, in practice, many substatements of a larger statement will only require witnesses held by a single party. These substatements could be extracted and proved using a regular \ac{cpnizk}, that does not suffer from the overhead caused by using \ac{MPC}.
\end{itemize}

\subsection{Practical PA-MPC from collaborative CP-NIZKs}\label{subsec:pa-mpc}
\emph{\Ac{PA-MPC}}~\cite{baum_publicly_2014} extends \ac{MPC} with a publicly-verifiable proof attesting to correct computation of the output. This correctness holds with respect to the public commitments to each party's inputs.
There exist several ways of instantiating this method for \ac{VPPC}.
However, most practical schemes are based on \acp{nizk}.
Especially those based on succinct \acp{nizk} are highly practical, due to their small communication and verification costs.
For an overview and discussion of \ac{PA-MPC} schemes we refer to~\cite{bontekoe_verifiable_2023}.

Ozdemir and Boneh~\cite{ozdemirExperimentingCollaborativeZkSNARKs2022} discuss how to construct \ac{PA-MPC} from collaborative proofs.
They use collaborative proofs to construct a proof for the \ac{MPC} computation result, whilst keeping the witness secret.
Since each party commits to their inputs as part of the zero-knowledge proof, their construction also requires the opening of a commitment inside the proof.

For general-purpose \acp{nizk} this is possible, however does lead to a significant increase in the number of constraints needed to encode the statement, thus leading to, e.g., increased proof generation times.
Especially for \iac{MPC} prover, this increase quickly becomes a bottleneck.
We observe that dedicated \ac{cpnizk} schemes are notably more efficient at proving statements including commitment openings.

Thus, the adoption of collaborative \acp{cpnizk} likely leads to significant performance improvements for \ac{PA-MPC}.
Even more so, when taking the techniques of \cref{subsec:composition} into account.
For formal definitions and security guarantees of \ac{PA-MPC} from collaborative proofs we refer to~\cite{ozdemirExperimentingCollaborativeZkSNARKs2022,bontekoe_verifiable_2023}.
\section{Distributed Protocols for CP-NIZKs} \label{sec:appr}
In this section, we describe our methods for constructing MPC protocols for two specific NIZKs: \emph{LegoGro16} and \emph{Bulletproofs}. Both can be used to construct \acp{cpnizk} as we will show.
We begin by describing two techniques that we use to construct the collaborative \acp{cpnizk}: collaborative commitment and \cplink. Our techniques are modular, and we strongly believe that they are more generally applicable to any NIZK that works over Pedersen-like commitments.

\subsection{Collaborative Pedersen-like commitments} \label{subsec:collab-comm}

\textbf{Overview of Pedersen-like commitments.}
Most \acp{cpnizk}, such as those introduced in LegoSNARK~\cite{campanelliLegoSNARKModularDesign2019} or based on Bulletproofs~\cite{bunz_bulletproofs_2018} rely upon Pedersen-like commitments.
These are commitment schemes whose verification algorithm follows the structure of the Pedersen vector commitment scheme~\cite{pedersen_non-interactive_1992}.
We recall the Pedersen vector commitment scheme, for message vectors of length $n$:

\begin{definition}[Pedersen vector commitment]
    Given a group \GG of order $p$, define input space $\domain = \ZZ^n_p$, commitment space $\mathcal{C} = \GG$, and opening space $\mathcal{O} = \ZZ_p$.
    \begin{itemize}[leftmargin=*,noitemsep]
        \item $\setup(\secparam) \rightarrow \ck = (g_0, \ldots g_n) \sample \GG^{n+1}$;
        \item $\commit(\ck, u, o) \rightarrow c = g_0^o \cdot \prod_{i=1}^{n} g_i^{u_i}$, with $o \sample \ZZ_p$;
        \item $\vercommit(\ck, c, u, o) \rightarrow c \overset{?}{=} g_0^o \cdot \prod_{i=1}^{n} g_i^{u_i}$.
    \end{itemize}
\end{definition}

The above scheme is perfectly hiding and computationally binding, given that the discrete log problem is hard in \GG.
We note that it is possible to open a Pedersen commitment within a non-CP-\ac{nizk} for arithmetic circuits, as done in, e.g.,~\cite{hopwoodZcashProtocolSpecification2023,bontekoeBalancingPrivacyAccountability2022}.
However, this is significantly less efficient than in \acp{cpnizk}~\cite{campanelliLegoSNARKModularDesign2019}.

\textbf{Our Protocol.}
One of the main distinguishing components of the CP-NIZK prover is generating a commitment to the witness. The collaborative CP-NIZKs we consider operate over data committed to with a Pedersen vector commitment and the private witness is distributed. Below, we present \iac{MPC} protocol that extends $\commit(\ck, u, o)$ to compute a single commitment from a shared witness vector $u = (u_1, \ldots, u_n)$, for simplicity we assume that party $i$ knows $u_i$, and thus $n=N$. In practice, the witness vector may hold more elements than the number of parties, and each party could hold an arbitrary number of these witnesses. The methods we present are trivially extended to that case.

% \begin{align*}
%   c \leftarrow \text{Ped.Commit}(\ck, \mathbf{w}, r)
% \end{align*}

% Where $c$ is the commitment, $\ck$ is the commitment key, and $r$ is a blinding factor. 
% Intuitively, there are two ways the parties can collaboratively commit to the witness input:
% \begin{enumerate}
%   \item Each party commits to their witness elements, the commitments are combined into a single commitment, and then the witness is shared and distributed among the parties (Commit-then-Share).
%   \item The witness vector is shared and distributed first, and each party commits to its shared version of the witness vector (Share-then-Commit).
% \end{enumerate}
% We will refer to these two approaches as ``Commit-then-Share'' and ``Share-then-Commit'', respectively.
Intuitively, since Pedersen commitments are additively homomorphic~\cite{pedersen_non-interactive_1992}, i.e., the product of two commitments is a commitment to the sum of the committed vectors, there are two ways the parties can collaboratively commit to the witness vector. We will refer to the approaches as \emph{\acl*{cts}} and \emph{\acl*{stc}}.

\emph{\ac{cts}.}
In the \ac{cts} approach, each party commits to their witness element $u_i$ before sharing. To keep the commitments hidden, each party must add a blinding factor $g_0^{o_i}$ to the commitment as follows:
\begin{align*}
  c_i & = g_0^{o_i} \cdot g_i^{u_i}
\end{align*}
Then, each party sends their commitment to the other party where the commitments are multiplied to obtain a single, final commitment to the full witness vector $u$: $c = g_0^{o'} \cdot \prod_{i=1}^{n} g_i^{u_i}$, where $o' = \sum_{i=1}^n o_i$.
As can be observed, the size of the commitment key is one group element per witness vector element, plus one common group element $g_0$ for the blinding factor, i.e., $n+1$ group elements in total.

\emph{\ac{stc}.}
The \ac{stc} approach involves sharing and distributing the witness vector $u$ among the parties before commitment. Initially, the witness vector, including a blinding factor, is split into shares and distributed among the parties. Then, each party commits to their shared version of the witness vector. Because the Pedersen commitment process is a linear operation, it works under the additive secret-sharing scheme. Concretely party $i$ computes its share as:
\begin{align*}
  \llbracket c \rrbracket_i & = g_{0}^{\llbracket o \rrbracket_i} \prod_{j=1}^{n} g_j^{\llbracket u_j \rrbracket_i}.
\end{align*}
Just like for \ac{cts}, the commitment key consists of $n+1$ group elements. Finally, the parties reveal the final full commitment by exchanging and multiplying the commitment shares: $c = \prod_{i=1}^n \llbracket c \rrbracket_i.$

\subsection[\texorpdfstring{Collaborative \cplink}{Collaborative CPlink}]{Collaborative \cplink}\label{subsec:collab-cplink}
In this section, we demonstrate the technique used to link the input commitment in \acp{cpnizk}, whose verification algorithm is the same as the Pedersen commitment, to an external commitment. This approach allows us to convert \iac{ccnizk} into a \acp{cpnizk}. To achieve this we rely on the \cplink construct from~\cite{campanelliLegoSNARKModularDesign2019}. \cplink provides a way to prove that two commitments, with different keys, open to the same vector $u$ (witness). Essentially, in this work, we use this technique to prove that the committed witnesses in two \acp{cpnizk} (LegoGro16 and Bulletproof) are equal. 

\cplink is built from \iac{zksnark} for linear subspaces $\Pi_{ss}$, which essentially provides a way to prove the relation $R_M(\mathbf{x}, \mathbf{w})$:
\[
R_M(\mathbf{x}, \mathbf{w}) = 1 \iff \mathbf{x} = \mathbf{M} \cdot \mathbf{w} \in \mathbb{G}_1^l,
\]
where $\mathbf{M} \in \mathbb{G}_1^{l \times t}$ is a public matrix, $\mathbf{x} \in \mathbb{G}_1^l$ a public vector, and $\mathbf{w} \in \ZZ_p^t$ a witness vector. The following are the key algorithms for the \ac{zksnark} for linear subspaces $\Pi_{ss}$:

\begin{itemize}[leftmargin=*]
    \item $\Pi_{ss}.\text{KeyGen}(\mathbf{M}) \rightarrow (\mathbf{ek}, \mathbf{vk})\colon \\ \quad \mathbf{k} \overset{\$}{\leftarrow} \ZZ_p^l, \, a \overset{\$}{\leftarrow} \ZZ_p; \quad \, \mathbf{P} \coloneq \mathbf{M}^{\top} \cdot \mathbf{k}; \quad \, \mathbf{C} \coloneq a \cdot \mathbf{k}$ \\
    $\quad \text{return } (\mathbf{ek} \coloneq \mathbf{P} \in \GG_1^l, \, \mathbf{vk} \coloneq (\mathbf{C'} = g_2^{\mathbf{C}}, a' = g_2^{a}) \in \GG_2^l \times \mathbb{G}_2)$ 
    \item $\Pi_{ss}.\text{Prove}(\mathbf{ek}, \mathbf{w}) \rightarrow \pi\colon \\ 
    \quad \text{return } \pi \leftarrow \mathbf{w}^{\top} \mathbf{P} \in \GG_1$
    \item $\Pi_{ss}.\text{VerProof}(\mathbf{vk}, \mathbf{x}, \pi) \rightarrow \{0, 1\}\colon \\ 
    \quad \text{ check that } \mathbf{x}^{\top} \cdot \mathbf{C'} \stackrel{?}{=} \pi \cdot a'$
\end{itemize}

Using $\Pi_{ss}$, we list the key algorithms for \cplink, which are simplified below for a setting with two commitments $c$ and $c'$, commitment keys $ck$ and $ck'$, and vector $u$: 
\begin{itemize}[leftmargin=*]
    \item $\mathsf{CP}_{\mathsf{link}}.\mathsf{KeyGen}(\mathbf{ck}, \mathbf{ck'}) \rightarrow (\mathbf{ek},\mathbf{vk})$: constructs the matrix $\mathbf{M} \leftarrow [g_0, 0, g_1, \ldots, g_n; 0, g'_0, g'_1, \ldots, g_n]$ for the relation $R_M$ using $\mathbf{ck} = (g_0, \ldots,g_n) \in \GG^{n+1}$ and $\mathbf{ck'}= (g'_0, \ldots, g'_n) \in \GG^{n+1}$ and generates evaluation and verification keys $(ek, vk) \leftarrow \Pi_{ss}.\text{KeyGen}(\mathbf{M}) $.
    \item $\cplink.\mathsf{Prove}(\mathbf{ek}, o, o', \mathbf{u}) \rightarrow \pi$: computes the proof $\pi \leftarrow \Pi_{ss}.\text{Prove}(\mathbf{ek},\mathbf{w})$ by using the evaluation key $\mathbf{ek}$, and setting $\mathbf{w} \leftarrow [o, o', u_1, \ldots, u_n]$ where $\mathbf{u} = (u_1, \ldots, u_n) \in \ZZ_p^n$.
    \item $\cplink.\mathsf{Verify}(\mathbf{vk}, c, c', \pi) \rightarrow \{0, 1\}$: verifies the proof by setting the vector $\mathbf{x} \leftarrow [c,c']$ in the subspace relation \(R_M\), and running the verification algorithm $\{0, 1\} \leftarrow \Pi_{ss}.\text{Verify}(\mathbf{vk}, \mathbf{x}, \pi)$
\end{itemize}

The \cplink algorithms can be extended to link multiple commitments as shown in~\cite{campanelliLegoSNARKModularDesign2019}. Our goal is to compute the prove function ($\Pi_{ss}.\text{Prove}$) using \ac{MPC}, where each party $i$ has a share $\llbracket \mathbf{w} \rrbracket_i$. We observe that the proof is essentially a multiplication of matrix $\mathbf{P}$ by the witness vector $\mathbf{w}$. Thus using additive secret sharing based \ac{MPC}, a simple approach is to run $\cplink.\mathsf{Prove}$ on each share of the witness, and then open the proof $\pi$. We also observe that $\Pi_{ss}$ is essentially a \acp{zksnark} that works on bilinear groups, therefore, we can apply the same optimization proposed in~\cite{ozdemirExperimentingCollaborativeZkSNARKs2022} to $\Pi_{ss}$. This allows us to combine \cplink with collaborative LegoGro16 and bulletproof efficiently, whilst also making them \acp{cpnizk}.
% We note that since \cplink operates in bilinear groups, CP-Bulletproof would require bilinear maps. Therefore, when implemented over pairing-friendly elliptic curves, it would result in slower performance.

\subsection{Collaborative LegoGro16}

\textbf{Overview of LegoGro16.}
This \ac{cpnizk} is first presented in~\cite{campanelliLegoSNARKModularDesign2019}, and is a commit-and-prove version of \emph{Groth16}~\cite{groth_size_2016}, a frequently used \ac{zksnark}.
Its proof consists of a regular \emph{Groth16} proof (which contains 2 elements from $\GG_1$ and one from $\GG_2$), plus one additional element $D$ (from $\GG_1$) that contains a commitment to the input.
\cplink guarantees that the element $D$ is a commitment to the same value as the external commitment $c$.

Proving a general statement about a commitment opening using \emph{LegoGro16} is around 5000$\times$ faster than using \emph{Groth16}, where the commitment is opened inside a regular Groth16 circuit.
Moreover, the \ac{CRS} is approximately 7000$\times$ shorter.
These advantages come at a very small cost.
Namely, the total proof size is slightly larger (191B versus 127B) and the verification time is around 1.2$\times$ slower.
These performance results further emphasize the potential of collaborative \acp{cpnizk} for realizing efficient \ac{PA-MPC}.

\textbf{Our Protocol.}
To realize the collaborative \emph{LegoGro16} protocol, we essentially combine the collaborative \emph{Groth16} protocol from~\cite{ozdemirExperimentingCollaborativeZkSNARKs2022} and the previous two protocols.
% the collaborative \cplink from \cref{subsec:collab-cplink}.
The resulting protocol involves the following steps:

\emph{Setup.} In the \setup phase, the \acp{CRS} for the \emph{Groth16} and \cplink protocols are generated.
The resulting keys are sent to all parties.

\emph{Commit and Distribute.} The input for the \commit algorithm is the commitment key \ck, and a vector of field elements (Witnesses).
Depending on whether the \ac{cts} or \ac{stc} approach (see \cref{subsec:collab-comm}) is used, the input vector will be either one party's part of the plain witness input or the shared witness input of all parties.
% Either way, the commitment algorithm works the same way.
% It expects a vector of field elements and is agnostic to what the elements represent.
The output is a commitment to the witness input. Then, in the \pcalgostyle{DistributeWitness} step of the protocol, all parties generate shares of their part of the witness and distribute the shares to all other parties.

% \emph{DistributeWitness.} In the \pcalgostyle{DistributeWitness} step of the protocol, all parties generate shares of their part of the witness and distribute the shares to all other parties.
% The input for the \pcalgostyle{DistributeWitness} algorithm is a vector of field elements representing the witness input of a party.
% The output is a vector of shared field elements representing the shared complete witness, where each party has a share of each element of the witness.

\emph{Prove.} The input for the collaborative prove algorithm \protocol is the proving key \pk, the circuit's public input, the shared witness input, \cplink's evaluation key $\mathbf{ek}$, and the opening to the external commitment.
This algorithm executes the \ac{MPC} protocol to generate a proof share for each party. The proof contains the standard \emph{Groth16} with the additional element $D$ and the \cplink proof. 

\emph{RevealProof.} The input for the \pcalgostyle{RevealProof} algorithm is a vector of all proof shares generated in the proving phase.
By exchanging the proof shares, the parties can reconstruct the full proof.

\subsection{Collaborative Bulletproofs for arbitrary arithmetic circuits} \label{sec:col-bp}
In what follows, we introduce collaborative Bulletproofs, starting with an overview of regular Bulletproofs~\cite{bunz_bulletproofs_2018}, followed by our collaborative version of Bulletproofs for arbitrary arithmetic circuits. As explained in \cref{subsec:collab-cplink}, this is then easily transformed into \iac{cpnizk} using \cplink. This does however warrant a minor modification to the verification algorithm (see \cref{app:bp-verify}).

\textbf{Overview of Bulletproofs.}
Bulletproofs have several advantages over other proof systems, including being based only on standard assumptions, not requiring bilinear maps, and having linear prover time complexity and proof size logarithmic in the number of constraints. Additionally, it allows building constraint systems on the fly, without a trusted setup. Bulletproofs are constructed using inner product arguments and employ a recursive approach for efficiency. This efficiency makes it feasible to use bulletproof in blockchains and confidential transactions.

Bulletproofs provide an efficient zero-knowledge argument for arbitrary arithmetic circuits, whilst also generalizing to include committed values (witnesses) as inputs to the arithmetic circuit. Including committed input wires is important as it makes Bulletproofs naturally suitable for \acp{cpnizk} without the need to implement an in-circuit commitment algorithm, thus aligning with our definitions.

\textbf{Our Protocol.}
We describe here how we constructed the collaborative bulletproof following the notation from~\cite{bunz_bulletproofs_2018}, for ease of comparison.
The prover begins by committing to its secret inputs $\mathbf{v} \in \ZZ_p^m$ using a blinding factor $\gamma \in \ZZ_p^m$ and generating a Pedersen commitment $\mathbf{V}$ where:
\[
V_j = g^{v_j} h^{\gamma_j} \quad \forall j \in [1, m]
\]
In our setting, $\mathbf{v}$ is distributed among the $N$ provers. Each prover $\mathcal{P}_i$ has a share $\llbracket \mathbf{v} \rrbracket_i$ of the witness. Therefore, all provers collaboratively generate this commitment using either the \ac{cts} or the \ac{stc} protocols as described earlier.

Subsequently, the provers build the constraint system, which allows them to perform a combination of operations to generate the constraints. Specifically, the \ac{R1CS} is used, which consists of two sets of constraints: (1) multiplication gates and (2) linear constraints in terms of the input variables.

(1) Multiplication gates simply take two input variables and multiply them to get an output. All multiplication gates in the constraint system can be expressed by the relation:
\[
 \mathbf{a}_L \circ \mathbf{a}_R = \mathbf{a}_O
\]

Where $\mathbf{a}_L \in \ZZ_p^n$ is the vector of the first input to each gate, $\mathbf{a}_R \in \ZZ_p^n$ is the vector of the second input to each gate, and $\mathbf{a}_O \in \ZZ_p^n$ is the vector of multiplication results. All three vectors have size $n$ representing the number of multiplication gates. 

(2) the input variables are used to express $Q$ Linear constraints in the form:
\[
\mathbf{W}_L \cdot \mathbf{a}_L + \mathbf{W}_R \cdot \mathbf{a}_R + \mathbf{W}_O \cdot \mathbf{a}_O = \mathbf{W}_V \cdot \mathbf{v} + \mathbf{c}
\]
Where $\mathbf{W}_L, \mathbf{W}_R, \mathbf{W}_O \in \ZZ_p^{Q \times n}$ are the public constraints matrices representing the weights applied to the respective inputs and outputs. $\mathbf{W}_V \in \ZZ_p^{Q \times m}$ is the matrix representing the weights for a commitment $\mathbf{V}$ and $\mathbf{c} \in \ZZ_p^Q$ is a constant public vector used in the linear constraints.

The provers' goal in collaborative bulletproofs is to jointly prove that there exists a $\mathbf{v}$ (in commitment $\mathbf{V}$) such that the linear constraints are satisfied while maintaining the secrecy of their respective $\llbracket \mathbf{v} \rrbracket_i$. The proof system can be summarized in the following relation.:
\[
\left\{
\begin{array}{l}
\left( 
\begin{array}{l}
\mathbf{g}, \mathbf{h} \in \GG^n, \mathbf{V} \in \GG^m, g, h \in \GG, \\
\mathbf{W}_L, \mathbf{W}_R, \mathbf{W}_O \in \ZZ_p^{Q \times n}, \mathbf{W}_V \in \ZZ_p^{Q \times m}, \\
\mathbf{c} \in \mathbb{Z}_p^Q; \mathbf{a}_L, \mathbf{a}_R, \mathbf{a}_O \in \ZZ_p^n, \mathbf{v}, \gamma \in \ZZ_p^m 
\end{array}
\right)\colon
\\
V_j = g^{v_j} h^{\gamma_j} \forall j \in [1, m] \land \mathbf{a}_L \circ \mathbf{a}_R = \mathbf{a}_O \\
\land \mathbf{W}_L \cdot \mathbf{a}_L + \mathbf{W}_R \cdot \mathbf{a}_R + \mathbf{W}_O \cdot \mathbf{a}_O = \mathbf{W}_V \cdot \mathbf{v} + \mathbf{c}
\end{array}
\right\}
\]
In addition to proving the previous relation, the provers utilize \cplink to prove that $\mathbf{V}$ commits to the same witnesses $\mathbf{v}$ as an external commitment $\hat{V}$. 

\textbf{Specification.}
In \cref{fig:protocol-col-bp}, we summarize the protocol for proving this relation collaboratively with $N$ provers. The protocol relies on $\Pi_{DBP}$, a sub-protocol for collaboratively generating proofs for arbitrary arithmetic circuits. $\Pi_{DBP}$ is our distributed construction of the regular bulletproofs from~\cite{bunz_bulletproofs_2018}. Unlike in~\cite{bunz_bulletproofs_2018}, we show $\Pi_{DBP}$ in its non-interactive form in \cref{fig:protocol-dbp-1}. We use the cryptographic hash function denoted as $\mathbf{H}$ to hash the transcript up to that point including the statements to be proven $\mathbf{st}$. The prover's input to the collaborative Bulletproof protocol in \cref{fig:protocol-col-bp} includes those required in the standard protocol, along with $\mathbf{ek}$ the \cplink evaluation key and $\llbracket\hat{o}\rrbracket$ the party's share of opening to the external commitment.

Additionally, the proving system employs the \ac{IPA} protocol by expressing all the constraints in terms of a single inner product and then running the \ac{IPA} protocol. This reduces the communication cost (proof size) since \ac{IPA} has logarithmic communication complexity. The IPA is not zero-knowledge itself~\cite{bunz_bulletproofs_2018}, nor does it need to be, and therefore can be executed individually by each party on the revealed vectors $(\mathbf{l}, \mathbf{r})$ as shown in step 11 of \cref{fig:protocol-dbp-1}. This approach does not require any communication between parties, however, all parties must check that the proof $\pi_{IPA}$ sent to the verifier is consistent with the one they generated locally. An alternative but costly approach is to run the IPA protocol in a distributed fashion. This ensures that all parties generate the same proof $\pi_{IPA}$. We present a distributed \ac{IPA} protocol in~\cref{app:ipa}, and experimentally compare its performance to the non-distributed version.

\begin{figure}[htbp]
\centering
\begin{mdframed}
\small
\textbf{$\mathcal{P}_i$'s input:} $(g, h \in \GG, \llbracket \mathbf{v} \rrbracket_i \in \ZZ_p^m, \mathbf{ek} \in \GG^{l}, \llbracket \hat{o} \rrbracket_i \in \ZZ_p)$

\textbf{Output:} $(\pi, \mathbf{V}, \pi_{\cplink})$ \\

\begin{itemize}[leftmargin=*]
\item[1.] \textbf{Generate input commitment:}
    \begin{itemize}[leftmargin=*]
        \item[a.] $\mathcal{P}_{i}{:}~\llbracket \mathbf{\gamma} \rrbracket_i \overset{\$}{\leftarrow} \ZZ_p^m$
        \item[b.] Parties use \ac{cts} or \ac{stc} to commit to input $\mathbf{v}$, each using their share of the input $\llbracket \mathbf{v} \rrbracket_i$ and blinding factors $\llbracket \mathbf{\gamma} \rrbracket_i$, resulting in commitment $\mathbf{V} \in \GG^m$:
            \begin{align*}
                V_j = g^{ v_j} h^{ \gamma_j } \quad \forall j \in [1, m]
            \end{align*}
    \end{itemize}

\item[2.] \textbf{Build constraint system:} 
    \begin{itemize}[leftmargin=*]
        \item[a.] Parties build the constraint systems for the required statements to be proven, generating a public description of the circuit $ \mathbf{st} = (g, h, \mathbf{g}, \mathbf{h}, \mathbf{W}_L, \mathbf{W}_R, \mathbf{W}_O, \mathbf{W}_V, \mathbf{c}, \mathbf{V})$ 
        \item[b.] Each party $\mathcal{P}_{i}$ assign the left and right inputs to each multiplication gate using their shares of $\llbracket \mathbf{v} \rrbracket_i$ and obtains ($\llbracket \mathbf{a}_L \rrbracket_i, \llbracket \mathbf{a}_R \rrbracket_i$).
        \item[c.] Parties use MPC to compute: $ \llbracket \mathbf{a}_O \rrbracket_i = \llbracket \mathbf{a}_L \rrbracket_i \circ \llbracket \mathbf{a}_R \rrbracket_i$
    \end{itemize}

\item[3.] \textbf{Prove:} Each $\mathcal{P}_{i}$ runs $\Pi_{DBP}$ and obtains the proof $\pi$
    \begin{align*}
        \pi \leftarrow \Pi_{DBP}(&g,h,\mathbf{g},\mathbf{h}, 
        \mathbf{c}, \llbracket \mathbf{a}_L \rrbracket_i, \llbracket \mathbf{a}_R \rrbracket_i, \llbracket \mathbf{a}_O \rrbracket_i, \\
        &\mathbf{W}_L, \mathbf{W}_R, \mathbf{W}_O, \mathbf{W}_V, \llbracket \gamma \rrbracket_i)
    \end{align*}
\item[4.] \textbf{Link:} 
    \begin{itemize}[leftmargin=*]
        % \item[a.] Each $\mathcal{P}_{i}$ computes: 
        %     \begin{align*}     
        %     \mathbf{ck} &\leftarrow (h, g) \\ (\llbracket \mathbf{ek} \rrbracket_i, \llbracket \mathbf{vk} \rrbracket_i) &= \mathsf{CP}_{\mathsf{link}}.\mathsf{KeyGen}(\mathbf{\mathbf{ck}}, \hat{\mathbf{ck}}) 
        %     \end{align*}
        % \item[b.] $\mathcal{P} :$ open ($\mathbf{ek}, \mathbf{vk}$)
        \item[a.] Each $\mathcal{P}_{i}$ computes: 
            \begin{align*}
                \llbracket \gamma' \rrbracket_i &= \sum_{j=1}^{m} \llbracket \gamma_j \rrbracket_i \\
                \llbracket \pi_{\cplink} \rrbracket_i &\leftarrow \cplink.\mathsf{Prove}(\mathbf{ek}, \llbracket \gamma' \rrbracket_i, \llbracket \hat{o} \rrbracket_i, \llbracket \mathbf{v} \rrbracket_i)
            \end{align*}
        \item[b.] $\mathcal{P}{:} $ open $\pi_{\cplink}$
    \end{itemize}
\item[5.] \textbf{Output:} ($\pi, \mathbf{V}, \pi_{\cplink}$)
\end{itemize}
\end{mdframed}
\caption{Collaborative bulletproof protocol for arbitrary arithmetic circuits}
\label{fig:protocol-col-bp}
\end{figure}

\begin{figure}[htbp]
\centering
\begin{mdframed}
\small
\textbf{$\mathcal{P}_i$'s input:}\\
\[
\begin{pmatrix}
g, h \in \mathbb{G}, \mathbf{g}, \mathbf{h} \in \mathbb{G}^n, 
\mathbf{c} \in \mathbb{Z}_p^Q, \llbracket \mathbf{a}_L \rrbracket_i, \llbracket \mathbf{a}_R \rrbracket_i, \llbracket \mathbf{a}_O \rrbracket_i \in \mathbb{Z}_p^n, \\ 
\mathbf{W}_L, \mathbf{W}_R, \mathbf{W}_O \in \mathbb{Z}_p^{Q \times n},
\mathbf{W}_V \in \mathbb{Z}_p^{Q \times m},
\llbracket \gamma \rrbracket_i \in \mathbb{Z}_p^m
\end{pmatrix}
\]

\textbf{Output:} proof $\pi$ 

\begin{itemize}[leftmargin=0.5cm]
    \item[1.] Each $\mathcal{P}_i$ computes: 
    \begin{align*}
    &\llbracket \alpha \rrbracket_i, \llbracket \beta \rrbracket_i, \llbracket \rho \rrbracket_i \overset{\$}{\leftarrow} \ZZ_p \\
    &\llbracket A_I \rrbracket_i = h^{\llbracket \alpha \rrbracket_i} \mathbf{g}^{\llbracket \mathbf{a}_L \rrbracket_i} \mathbf{h}^{\llbracket \mathbf{a}_R \rrbracket_i} \in \GG \\
    &\llbracket A_O \rrbracket_i = h^{\llbracket \beta \rrbracket_i} \mathbf{g}^{\llbracket \mathbf{a}_O \rrbracket_i} \in \GG \\
    &\llbracket \mathbf{s}_L \rrbracket_i, \llbracket \mathbf{s}_R \rrbracket_i \overset{\$}{\leftarrow} \ZZ_p^n \\
    &\llbracket S \rrbracket_i = h^{\llbracket \rho \rrbracket_i} \mathbf{g}^{\llbracket \mathbf{s}_L \rrbracket_i} \mathbf{h}^{\llbracket \mathbf{s}_R \rrbracket_i} \in \GG
    \end{align*}
    \item[2.] $\mathcal{P}{:} $ open $(A_I, A_O, S)$
    \item[3.] $\mathcal{P}{:}~y \leftarrow \mathbf{H}(\mathbf{st}, A_I, A_O, S), z \leftarrow \mathbf{H}(A_I, A_O, S, y) \in \mathbb{Z}_p^*$
    \item[4.] Each $\mathcal{P}_i$ computes:
        \begin{align*}
            &\mathbf{y}^n = (1, y, y^2, \ldots, y^{n-1}) \in \ZZ_p^n \\
            &\mathbf{z}^{Q+1}_{[1:]} = (z, z^2, \ldots, z^Q) \in \ZZ_p^Q \\
            &\delta(y, z) = \langle \mathbf{y}^{-n} \circ (\mathbf{z}^{Q+1}_{[1:]} \cdot \mathbf{W}_R), \mathbf{z}^{Q+1}_{[1:]} \cdot \mathbf{W}_L \rangle
        \end{align*}
    \item[5.] Each $\mathcal{P}_i$ computes: 
        \begin{align*}
            \llbracket l(X) \rrbracket &= \llbracket \mathbf{a}_L \rrbracket \cdot X + \llbracket \mathbf{a}_O \rrbracket \cdot X^2 \\ 
            &+ \mathbf{y}^{-n} \circ (\mathbf{z}^{Q+1}_{[1:]} \cdot \mathbf{W}_R) \cdot X \\ 
            &+ \llbracket \mathbf{s}_L \rrbracket \cdot X^3 \in \ZZ^n_p[X] \\
            \llbracket r(X) \rrbracket &= \mathbf{y}^n \circ \llbracket \mathbf{a}_R \rrbracket \cdot X - \mathbf{y}^n \\
            &+ \mathbf{z}^{Q+1}_{[1:]} \cdot (\mathbf{W}_L \cdot X + \mathbf{W}_O) \\ 
            &+ \mathbf{y}^n \circ \llbracket \mathbf{s}_R \rrbracket \cdot X^3 \in \ZZ^n_p[X] \\
            \llbracket t(X) \rrbracket_i &= \langle \llbracket l(X) \rrbracket_i, \llbracket r(X) \rrbracket_i \rangle \\
            &= \sum_{j=1}^{6} \llbracket t_j \rrbracket_i \cdot X^j \in \ZZ_p[X] \\
            \llbracket \tau_j \rrbracket_i &\overset{\$}{\leftarrow} \ZZ_p \quad \forall j \in [1, 3, 4, 5, 6] \\ 
            \llbracket T_j \rrbracket_i &= g^{\llbracket t_j \rrbracket_i} \cdot h^{\llbracket \tau_j \rrbracket_i} \quad \forall j \in [1, 3, 4, 5, 6]
        \end{align*}
    \item[6.] $\mathcal{P}{:} $ open $ (T_j \quad \forall j \in [1, 3, 4, 5, 6])$
    \item[7.] $\mathcal{P}{:}~x \leftarrow \mathbf{H}(y, z, T_1, T_3, T_4, T_5, T_6) \in \mathbb{Z}_p^*$
    \item[8.] Each $\mathcal{P}_i$ computes:
    \begin{align*}
    \llbracket \mathbf{l} \rrbracket_i &= \llbracket l(x) \rrbracket_i \in \ZZ_p^n \\
    \llbracket \mathbf{r} \rrbracket_i &= \llbracket r(x) \rrbracket_i \in \ZZ_p^n \\
    \llbracket \tau_x \rrbracket_i &= \sum_{j=1, j \neq 2}^{6} \llbracket \tau_j \rrbracket_i \cdot x^j \\
    &+ x^2 \cdot \left( \mathbf{z}^{Q+1}_{[1:]} \cdot \mathbf{W}_V \cdot \llbracket \gamma \rrbracket_i \right) \in \mathbb{Z}_p \\
    \llbracket \mu \rrbracket_i &= \llbracket \alpha \rrbracket_i \cdot x + \llbracket \beta \rrbracket_i \cdot x^2 + \llbracket \rho \rrbracket_i \cdot x^3 \in \mathbb{Z}_p 
    \end{align*}
    \item[9.] $\mathcal{P}{:} $ open ($\mathbf{l}, \mathbf{r}, \tau_x, \mu$) 
    \item[10.] $\mathcal{P}{:}~x_u \leftarrow \mathbf{H}(x, \tau_x, \mu) \in \mathbb{Z}_p^*$
    \item[11.] Each $\mathcal{P}_i$ computes: 
    \begin{align*}
    \hat{t} &= \langle \mathbf{l}, \mathbf{r} \rangle \in \mathbb{Z}_p \\ 
    \pi_{IPA} &\leftarrow \Pi_{IPA}.\text{Prove}(\mathbf{g}, \mathbf{h}, g^{x_u}, \mathbf{l}, \mathbf{r})
    \end{align*}
    \item[12.] \textbf{Output:} \\ $\pi = (A_I, A_O, S, T_1, T_3, T_4, T_5, T_6, \tau_x, \mu, \hat{t}, \pi_{IPA})$ 
\end{itemize}
\end{mdframed}
\caption{Sub-protocol $\Pi_{DBP}$ for collaboratively generating bulletproofs for arbitrary arithmetic circuits}
\label{fig:protocol-dbp-1}
\end{figure}

\section{Implementation} \label{sec:impl}
To evaluate and experiment with our newly developed collaborative \acp{cpnizk}, a proof-of-concept system was implemented that allows multiple parties to generate collaborative \ac{cpnizk} proofs.

Since implementing collaborative \ac{cpnizk} relies heavily on finite fields and pairing-friendly curves, we used a library that abstracts the underlying cryptographic operations and offers a generic interface for working with finite fields and pairing-based cryptography.

Our implementation\footnote{Our open-source implementation will be made available} is built on top of Arkworks~\cite{noauthor_textttarkworks_2022}, a Rust ecosystem for \ac{zksnark} programming. It provides libraries (crates) for working with finite fields and elliptic curves and several implementations of existing \acp{zksnark} and other cryptographic primitives for implementing custom \acp{nizk}.

\textbf{MPC.} A number of prior works~\cite{ozdemirExperimentingCollaborativeZkSNARKs2022, garg_zksaas_2023, chiesa_eos_2023} also base their solutions on Arkworks. For instance, \cite{ozdemirExperimentingCollaborativeZkSNARKs2022} implements collaborative \acp{zksnark} using the Arkworks library, including the \acp{MPC} primitives. However, extending their implementation to support the \acp{cpnizk} considered in this work proved to be a difficult task. The implementation in \cite{ozdemirExperimentingCollaborativeZkSNARKs2022} uses version v0.2.0 of Arkworks, and since Arkworks is in active development, various breaking changes have occurred since. As a result, we opted to implement the MPC primitives using the recent version (v0.4.x). Although the implementation is made from scratch, some design decisions from these existing implementations are used. Similar to~\cite{ozdemirExperimentingCollaborativeZkSNARKs2022}, introducing an \acs{MPC} Pairing Wrapper is our primary implementation strategy, in addition to defining interfaces for shared field and group types. However, we limited our \ac{MPC} implementation to the SPDZ framework.

\textbf{Networking.} Our collaborative \ac{cpnizk} protocol requires communication between all parties. We implemented the \texttt{network} module to manage this communication. There are three types of communication between the parties: (1) exchanging witness shares; (2) communicating during the proving protocol to multiply two shared values; and (3) exchanging proof shares to get the final proof.

\textbf{Collaborative LegoGro16.} We based our LegoGro16 implementation on that of~\cite{campanelliLegoSNARKModularDesign2019} and updated it to be compatible with our version of Arkworks. Additionally, we observed that the same MPC optimization applied to Groth16 in~\cite{ozdemirExperimentingCollaborativeZkSNARKs2022} can also be applied to its CP counterpart (LegoGro16), and therefore, these optimizations were implemented. The LegoGro16 implementation in this work does not support parallelization, and the existing multi-scalar multiplication implementation is replaced with a single-threaded alternative. Parallelization is disabled because the order in which the parties exchange and reveal values is essential.
% To ensure that all parties always work on the same values, they must send and receive them in the same order. 
This problem could be solved by using a more complex communication protocol, but this is beyond the scope of this project. This decision is made similarly in prior works like~\cite{ozdemirExperimentingCollaborativeZkSNARKs2022}.

\textbf{Collaborative CP-Bulletproofs.} Our starting point for the Bulletproof implementation is an existing Bulletproof implementation~\cite{dalek_bulletproofs} that works over \emph{Curve25519} and is implemented in Rust. We adapted this implementation to support distributed provers by modifying the underlying field and group elements to work over the \ac{MPC} shared fields and group types we defined. Then, we applied the techniques described in \cref{sec:col-bp} by building a wrapper for the prover functions to allow multiple parties to generate a proof. We use Merlin transcripts \cite{dalek_merlin} to manage the generation of verifier challenges, ensuring that all provers have a consistent transcript to generate the correct challenge.

\textbf{Summary.} The main modules for the implementation can be summarized as follows:
\begin{itemize}[leftmargin=*]
    \item The \texttt{mpc} module contains all the functionality for implementing the MPC primitives on Arkworks;
    \item The \texttt{network} module implements all the networking and communication functionalities required for the parties to communicate with each other;
    \item The \texttt{col-LegoGro16} module contains the implementation for the collaborative LegoGro16;
    \item The \texttt{col-cp-BP} module contains the implementation for the collaborative CP-Bulletproofs.
\end{itemize}

\section{Experiments} \label{sec:eval}
Below, we present the main results of our extensive experimentation for evaluating the performance of our collaborative \acp{cpnizk} constructions. Additional experimental results are provided in \cref{apx-eval}. Our evaluation was conducted on a consumer machine equipped with a 10-core Apple M2 Pro CPU and 16GB of RAM. We focused on two metrics: single-threaded runtime and communication cost. We report the average results from three runs. We vary the following parameters to examine their impact on performance:
\begin{itemize}[leftmargin=*]
    \item The number of \ac{R1CS} constraints, varying from $2$ to $2^{15}$;
    \item The number of parties, increasing from $2$ to $2^{6}$.
\end{itemize}

With our experiments we answer the following questions:
\begin{enumerate}[leftmargin=*]
    \item \emph{How does the performance of both collaborative \acp{cpnizk} (LegoGro16 and Bulletproofs) compare to the single prover setting?}
    \item \emph{What is the overhead of composability (linking input commitments), i.e., how does collaborative \acp{cpnizk} compare to the non-CP variant such the ones in~\cite{ozdemirExperimentingCollaborativeZkSNARKs2022}?}
    \item \emph{Does the performance of using two different \acp{cpnizk} improves efficiency over using a single NIZK?}
\end{enumerate}

\begin{figure*}[htbp]
    \centering
    \begin{minipage}{0.32\textwidth}
        \centering
        \includegraphics[width=\textwidth]{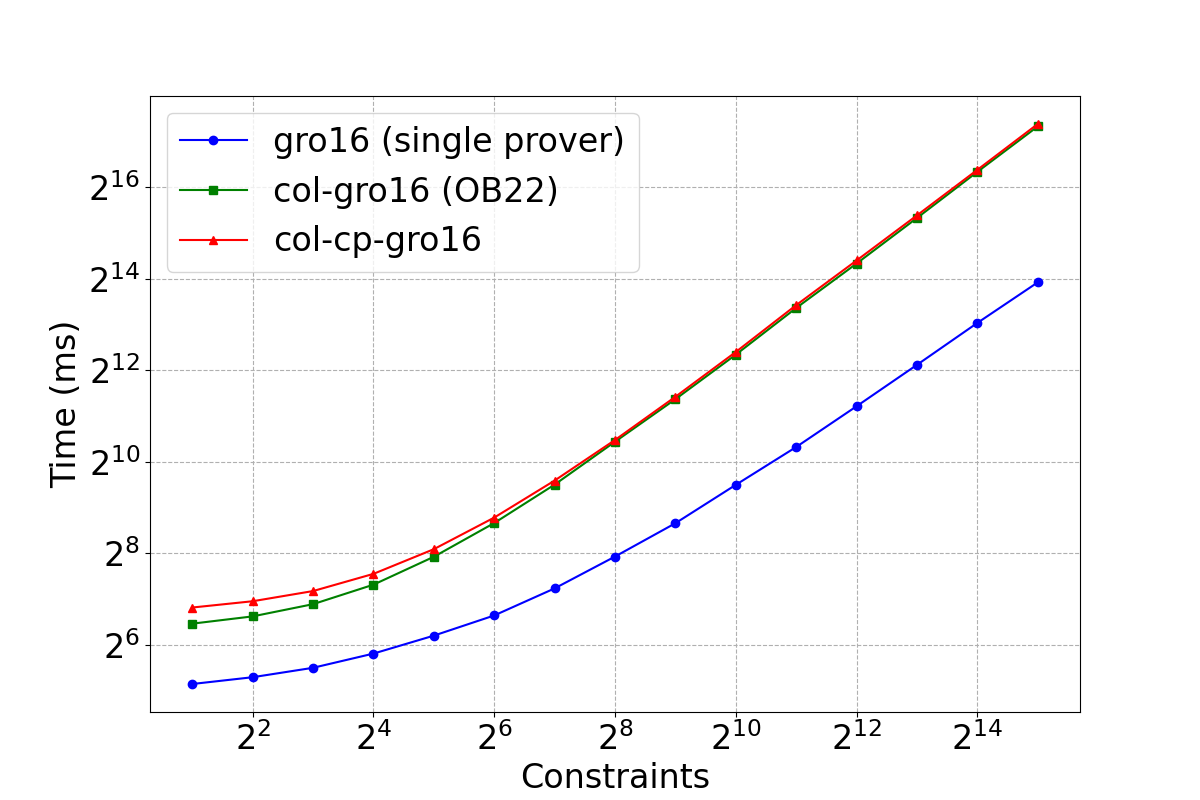}
        \caption{Runtime per prover party for bulletproofs: (1) Collaborative with \cplink (col-cp-bp), (2) Collaborative (col-bp), (3) Single prover (bp).}
        \label{fig:col-bp}
    \end{minipage}
    \hfill
    \begin{minipage}{0.32\textwidth}
        \centering
        \includegraphics[width=\textwidth]{img/gro16_const_fig.png}
        \caption{Runtime per prover party for: (1) Collaborative LegoGro16 (col-cp-gro16), (2) Collaborative Groth16 (col-gro16), (3) Single prover Groth16 (gro16).}
        \label{fig:col-gro16}
    \end{minipage}
    \hfill
    \begin{minipage}{0.32\textwidth}
        \centering
        \includegraphics[width=\textwidth]{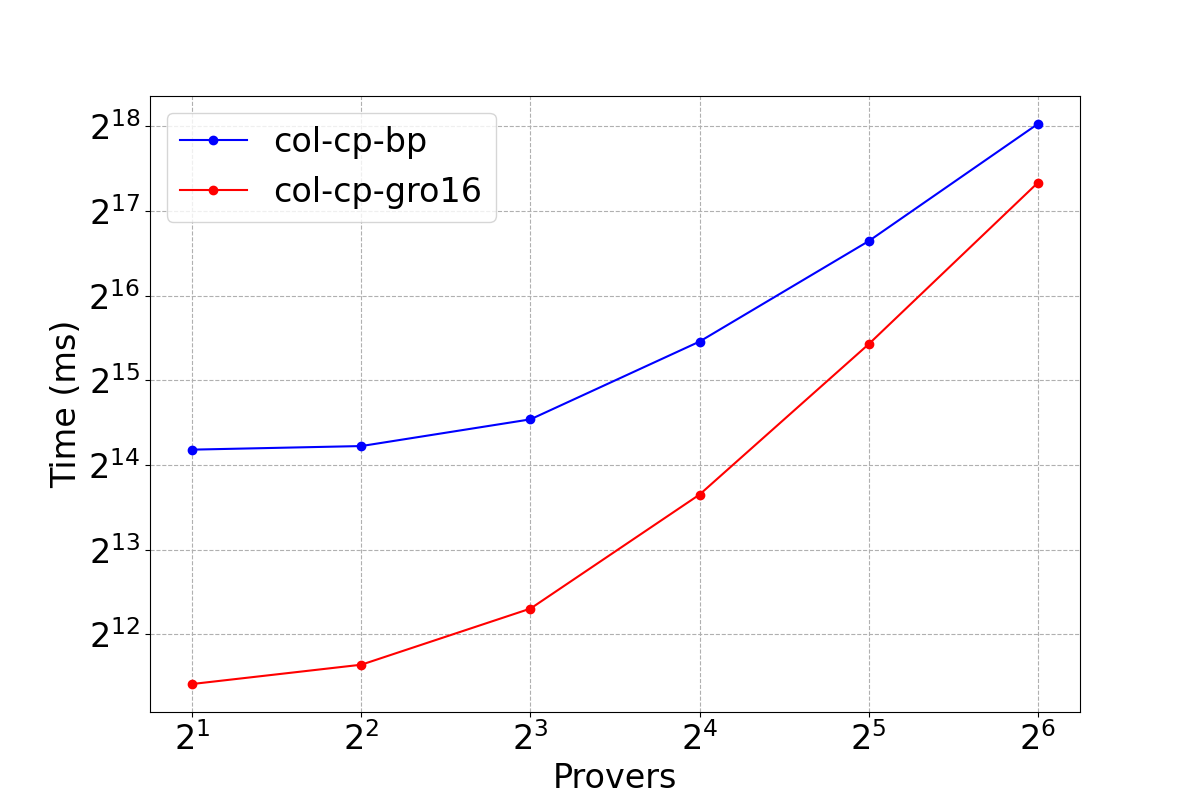}
        \caption{Runtime per prover party for varying number of constraints and prover group sizes.}
        \label{fig:party}
    \end{minipage}
    \hfill
    \begin{minipage}{0.32\textwidth}
        \centering
        \includegraphics[width=\textwidth]{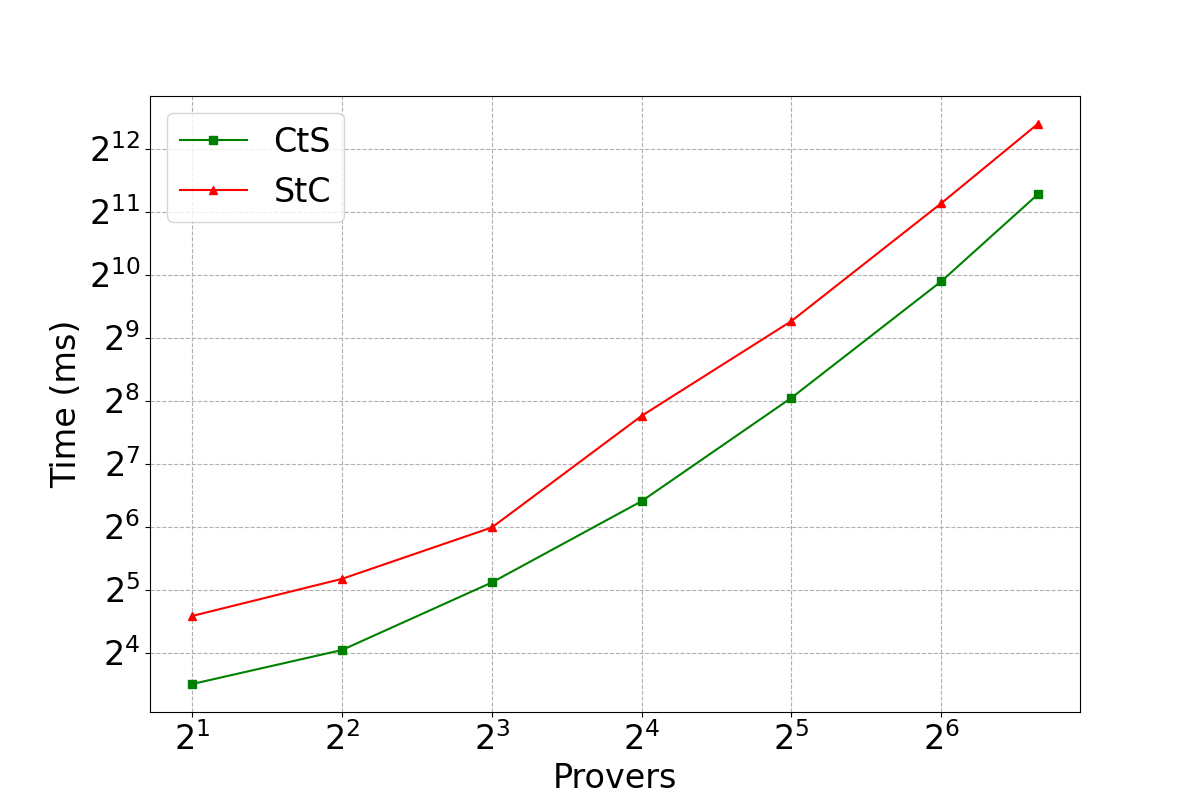}
        \caption{Runtime per prover party for \ac{cts} and \ac{stc} where each party provides a witness to the commitment.}
        \label{fig:col-commit}
    \end{minipage}
    \hfill
    \begin{minipage}{0.32\textwidth}
        \centering
        \includegraphics[width=\textwidth]{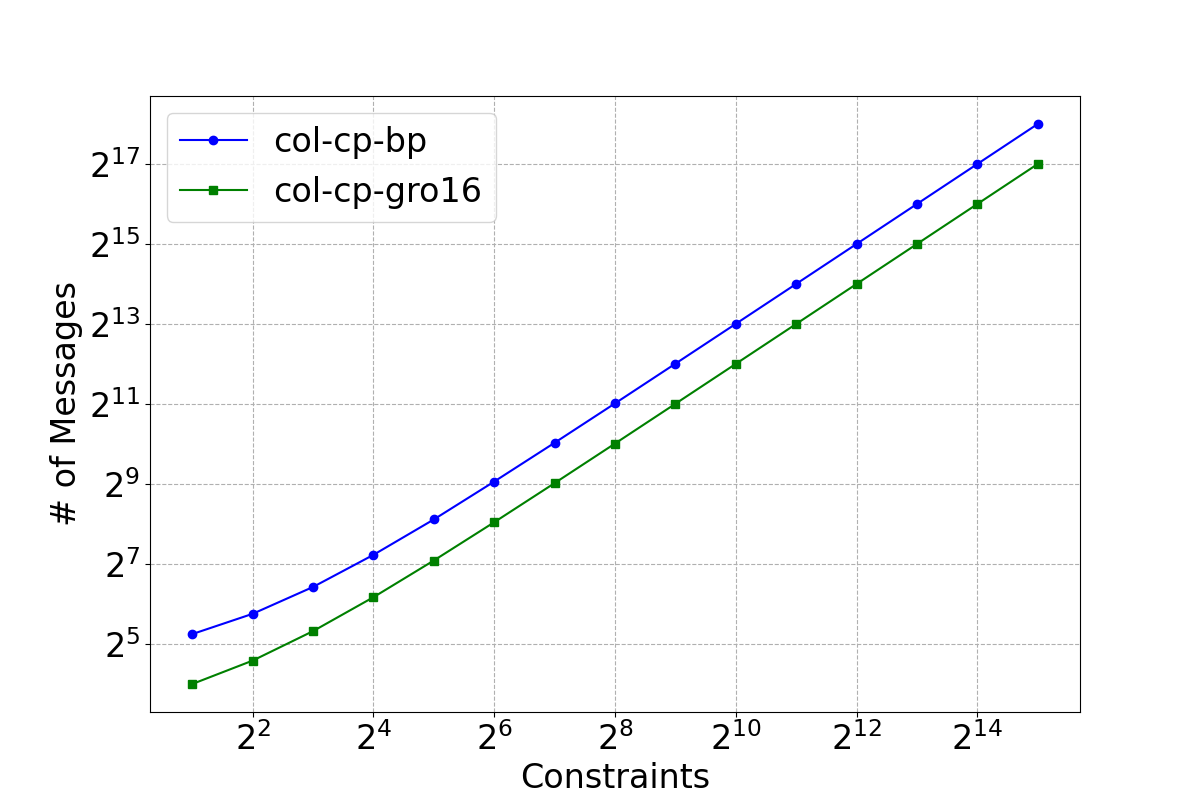}
        \caption{Collaborative LegoGro16 and bulletproofs communication costs for varying circuit sizes. The number of provers is $2$.}
        \label{fig:communication}
    \end{minipage}
    \hfill
    \begin{minipage}{0.32\textwidth}
        \centering
        \includegraphics[width=\textwidth]{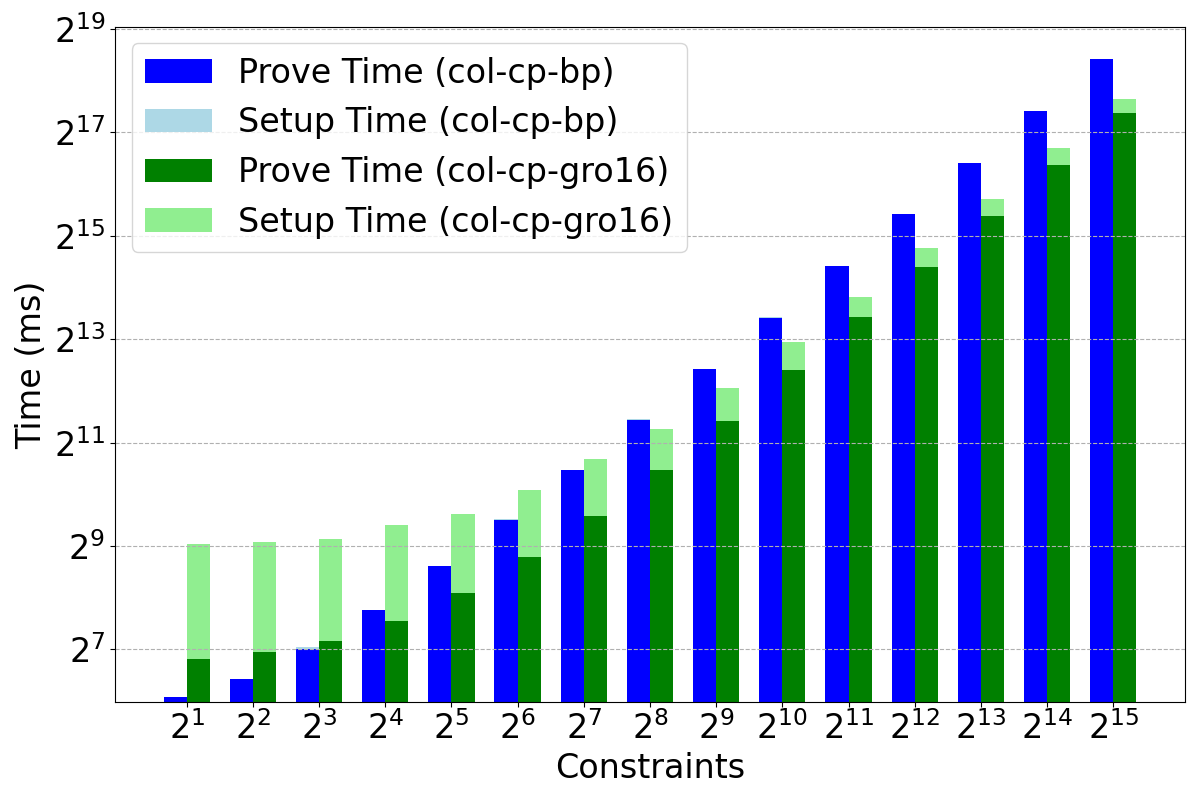}
        \caption{Setup and prove runtime per prover party for Collaborative LegoGro16 (col-cp-gro16) and Collaborative bulletproofs with \cplink (col-cp-bp). The number of provers is $2$.}
        \label{fig:compare}
    \end{minipage}
    \label{fig:combined}
\end{figure*}

\subsection{Setup and commitments}
The focus of this work is on distributing the prover in \acp{cpnizk} and evaluating the benefits of composability in this setting. Therefore, we omit the evaluation of the setup required for LegoGro16 to generate the \ac{CRS} and any computations necessary to prepare the witness and public input to the circuit. This decision is common for this line of work, since these computations heavily depend on the type of circuit used and can be performed when spare computational resources and bandwidth are available. Additionally, we omit evaluating the verifier cost since we do not modify the verification algorithms and these costs are thus unchanged.

An essential part of the setup in our construction is that provers must collaboratively generate a commitment to the witness, which can then be linked to an external commitment using \cplink. \cref{fig:col-commit} shows the performance of generating the commitment collaboratively using both sub-protocols: \ac{cts} and \ac{stc}. The performance is measured with an increasing number of provers $N$, where we assume each prover has a distinct witness. The communication cost for both sub-protocols is the same, consisting of $N-1$ broadcast messages per prover. However, the total runtime for the \ac{cts} method is significantly less than the \ac{stc} method, and this difference increases as the number of provers increases. This can be explained by the fact that the \ac{cts} method only requires each party to commit to their part of the witness, while the \ac{stc} method requires each party to commit to the shared witness, a vector of $N$ elements.

% \begin{figure}[htbp]
% \centering
% \includegraphics[width=0.40\textwidth]{img/col_com_plot.png}
% \caption{Runtime per prover party for \ac{cts} and \ac{stc} where each party provides a witness to the commitment.}
% \label{fig:col-commit}
% \end{figure}
% \vspace{-0.2cm}

\subsection{Varying number of constraints}
To demonstrate the scalability of collaborative \acp{cpnizk}, we evaluated both LegoGro16 and Bulletproof under varying numbers of constraints. The results are depicted in \cref{fig:col-bp,fig:col-gro16}, along with the performance of a single prover and non-collaborative variants of these NIZKs. The results show that even in distributed settings collaborative \acp{cpnizk} overhead is minimal and almost negligible in circuits with a large number of constraints, especially when compared to non-CP counterparts, such as the ones in~\cite{ozdemirExperimentingCollaborativeZkSNARKs2022}.

% \begin{figure}[htbp]
% \centering
% \includegraphics[width=0.40\textwidth]{img/bp_const.png}
% \caption{Runtime per prover party for bulletproofs: (1) Collaborative with \cplink (col-cp-bp), (2) Collaborative (col-bp), (3) Single prover (bp).}
% \label{fig:col-bp}
% \end{figure}

% \begin{figure}[htbp]
% \centering
% \includegraphics[width=0.40\textwidth]{img/gro16_const_fig.png}
% \caption{Runtime per prover party for: (1) Collaborative LegoGro16 (col-cp-gro16), (2) Collaborative Groth16 (col-gro16), (3) Single prover Groth16 (gro16).}
% \label{fig:col-gro16}
% \end{figure}

\subsection{Varying number of provers}
In our second experiment, we fixed the number of constraints at $2^{10}$ and varied the number of provers to show how collaborative \acp{cpnizk} scale with an increasing number of provers. We show the results in \cref{fig:party}. Our results are consistent with that in~\cite{ozdemirExperimentingCollaborativeZkSNARKs2022} for Groth16. For bulletproofs the increase in the number of provers causes a significant decrease in performance. This is mainly because bulletproofs requires more communication compared to Groth16 as will be shown in \cref{subsec:communication-costs}.

As can be observed from our results, the performance of running two instances of collaborative \acp{cpnizk} with $N$ provers compared to running one instance with $2N$ provers would result in $1.3$--$2\times$ improved performance, highlighting the importance of composition. Splitting prover groups becomes beneficial in settings with a large number of provers (${>}8$ provers). For instance, splitting 64 provers into 2 collaborative \acp{cpnizk} would improve the per-party runtime by ${\approx}~2\times$. 

% \begin{figure}[htbp]
% \centering
% \includegraphics[width=0.40\textwidth]{img/party_plot.png}
% \caption{Runtime per prover party for different prover group sizes. The number of constraints is $2^{10}$.}
% \label{fig:party}
% \end{figure}

\subsection{Communication Cost}\label{subsec:communication-costs}
In \cref{fig:communication}, we report the communication cost per party for both collaborative \acp{cpnizk} (LegoGro16 and bulletproof). The results show that bulletproof requires more communication than LegoGro16. 

% \begin{figure}[htbp]
% \centering
% \includegraphics[width=0.40\textwidth]{img/communication_plot.png}
% \caption{Collaborative LegoGro16 and bulletproofs communication costs for varying circuit sizes. The number of provers is $2$.}
% \label{fig:communication}
% \end{figure}

\subsection{Improving efficiency by composition} \label{subsec:eval-compare}
In our final experiment, we aim to answer our third question: whether using different collaborative \acp{cpnizk} would improve performance (for provers) compared to simply using a single general-purpose NIZK, such as those used in~\cite{ozdemirExperimentingCollaborativeZkSNARKs2022}. To answer this question, we first examined when using collaborative Bulletproofs or LegoGro16 would be less costly. For this, we plotted the runtime performance overhead of both protocols as we varied the number of constraints. \cref{fig:compare} shows our results. Based on these results, it can be observed that it is optimal to use LegoGro16 for large circuits (with a large number of constraints) and Bulletproofs for smaller circuits. Therefore, we conclude that it is indeed more efficient to use a combination of these two \acp{cpnizk} on the same input than to rely solely on a single one.
In \cref{sec:eval-app} we further support this conclusion by evaluating the performance improvement as a result of proof composition by means of an application scenario.

% \begin{figure}[htbp]
% \centering
% \includegraphics[width=0.40\textwidth]{img/compare.png}
% \caption{Setup and prove runtime per prover party for Collaborative LegoGro16 (col-cp-gro16) and Collaborative bulletproofs with \cplink (col-cp-bp). The number of provers is $2$.}
% \label{fig:compare}
% \end{figure}
% \vspace{-0.2cm}

\section{Application: private audits} \label{sec:eval-app}
As we discussed earlier in \cref{subsec:example-applications}, collaborative \acp{cpnizk} can be utilized to generate proofs in a setting where a mortgage applicant needs to demonstrate compliance with the bank's requirements. This setting is analogous to the proofs about net assets presented in~\cite{ozdemirExperimentingCollaborativeZkSNARKs2022}. Below, we discuss how collaborative \acp{cpnizk} can be constructed in this setting, evaluate the time required to construct the proof, and show how we significantly improve efficiency over the Collaborative SNARKs approach in~\cite{ozdemirExperimentingCollaborativeZkSNARKs2022}.

Consider $k$ transactions distributed among $N$ banks, where each transaction is a $64$-bit signed integer. Each bank publishes a Merkle tree commitment of all transactions. The applicant is required to prove that the net of their debits and credits across all banks exceeds a threshold $T$. This claim can be split into the following checks:
\begin{itemize}[leftmargin=*]
    \item[1.] The transactions are in the committed Merkle tree.
    \item[2.] The transactions are well-formed (represented as a $64$-bit signed integer).
    \item[3.] The sum of all transactions is computed correctly.
    \item[4.] The computed sum is larger than the threshold $T$.
\end{itemize}

A straightforward approach is to include all checks in one collaborative zk-SNARK, requiring around $364 \cdot n$ as demonstrated in~\cite{ozdemirExperimentingCollaborativeZkSNARKs2022}. However, since the claim can be split into multiple checks, we can utilize the modularity of our proposed collaborative \acp{cpnizk} and use a combination of both collaborative Groth16 and Bulletproofs to gain efficiency. In this setting, checking that each transaction is well-formed and is included in the Merkle tree can be done individually by each bank, as each bank can generate such proofs for their set of transactions. Then, computing the sum and checking it against the threshold $T$ can be done collaboratively using our construction of collaborative \acp{cpnizk}, which would link this collaborative proof with the individually generated proofs. We exploit this modularity even further by employing Groth16 for large circuits such as checks 1 and 4, and Bulletproofs for checks 2 and 3.

For evaluating the two approaches (on our consumer machine) we set the number of transactions to 90 (${\approx}~2^{15}$ constraints) and the number of banks to 2. Generating a proof would take ${\approx}~165$ seconds and require ${\approx}~2^{17}$ messages exchanged using the approach in~\cite{ozdemirExperimentingCollaborativeZkSNARKs2022}. Using our construction, the two parties each with 45 transactions will run $2^{14}$ constraints on their set of transactions locally (without MPC) using legoGro16 in ${\approx}~8.3$ seconds. Then the two banks will collaboratively run \acp{cpnizk} to compute the sum and compare to T requiring $64 + log_2 k \approx 71$ constraints in ${\approx}~0.5$ seconds. Thus, generating the proof using our construction is done in ${\approx}~9$ seconds and requires ${\approx}~2^8$ messages exchanged. This results in an $18 \times$ improvement in per-party runtime and only $0.2\%$ of the required communication.
We note that if the number of transactions is fixed at 90 but the number of parties (banks) increases, the per-party runtime decreases since each party will locally prove $k/N$ transactions ($2^{15}/N$ constraints). We evaluated this for 4 and 8 banks, finding that the improvement in runtime becomes $33\times$ and $55\times$, respectively. Overall, our protocol is $18$--$55\times$ faster in this application setting compared to~\cite{ozdemirExperimentingCollaborativeZkSNARKs2022} with only a fraction of the required communication between parties.
We estimate that this improvement in performance increases with a higher number of transactions (constraints). Processing $2^{20}$ constraints using~\cite{ozdemirExperimentingCollaborativeZkSNARKs2022} approach would be infeasible on a consumer machine, taking ${\approx}~400$ minutes. In contrast, using our approach the time taken would be ${\approx}~9.5$ minutes. Therefore, we project that for $2^{20}$ constraints, the performance enhancement would be in the range $40$--$200\times$.

\section{Conclusion}\label{sec:conclusion}
We have formally defined collaborative \acp{cpnizk} and showed how to generically construct these from existing \acp{nizk} and \acp{MPC} frameworks, by taking advantage of the modularity of our construction.
Moreover, we present two commitment paradigms that allow the adoption of \acp{cpnizk} in varying settings, both adaptive and on-the-fly.

By combining the strengths of collaborative proofs with the composability of the commit-and-prove paradigm, we achieve significant efficiency improvements over non-composable counterparts while maintaining the flexibility of collaborative proofs. Our implementation in Arkworks demonstrates the practical feasibility and modularity of our approach, paving the way for future work and extensions to easily integrate novel \acp{nizk} and \acp{MPC} schemes. Our experiments show that our approach incurs minimal overhead and provides substantial performance benefits, including efficiency improvements of $18$--$55\times$ compared to existing constructions in realistic application scenarios. This makes collaborative \acp{cpnizk} a valuable tool for a wide range of applications.

\bibliographystyle{IEEEtranS}
\bibliography{ref}

\appendix
\subsection{Additional Related Work} \label{apx-additional-relatedwork}
\textbf{PA-MPC from distributed proofs.}
The research on \ac{PA-MPC}~\cite{baum_publicly_2014} is closely related to the context of our work.
However, constructions for \ac{PA-MPC} are often very specific to one scheme and do not consider composability, contrary to our generic, modular approach.
In \ac{PA-MPC}, a group of parties, each holding some private values, collaboratively evaluates a function, in such a way that correctness of the result can be verified by any external party.
\Ac{PA-MPC} is often proposed for the outsourcing setting, in which data parties secret share there data to multiple servers, who compute the proof in a distributed fashion.
Since data parties themselves can also be computation servers, it can also be used in the setting of \acp{VPPC}~\cite{bontekoe_verifiable_2023}.
Correctness is, generally, derived from commitments to the parties' inputs.

In the initial construction for \ac{PA-MPC}~\cite{baum_publicly_2014} a verifier had to assert correctness of a full transcript.
However, later it was shown by Veeningen that \ac{PA-MPC} can be made more efficient by creating a distributed \ac{ZKP} on committed data.
Veeningen showed how to construct \ac{PA-MPC} from an adaptive \ac{cpsnark} based on Pinocchio~\cite{parno_pinocchio_2013}.
Here, adaptive means that the scheme is secure even when the relation is chosen after the input commitments are generated.
\Ac{PA-MPC} is realized by applying the \ac{cpsnark} approach to Trinocchio~\cite{schoenmakers_trinocchio_2016}, a distributed version of~\cite{parno_pinocchio_2013}.

Kanjalkar et al.~\cite{kanjalkar_publicly_2021} improve upon~\cite{veeningen_pinocchio-based_2017}, by constructing \ac{PA-MPC} from a new \acp{cpsnark} based on Marlin~\cite{chiesa_marlin_2019}, which only requires a single trusted setup ceremony that can be used for any statement, rather than the relation-specific trusted setup of Pinocchio.
More recently, Dutta et al.~\cite{dutta_compute_2022} consider the notion of \emph{authenticated \ac{MPC}}, on signed inputs.
They construct several distributed proofs of knowledge for opening Pedersen commitments and compressed $\Sigma$-protocols~\cite{attemaCompressedSProtocolTheory2020}.

\subsection{Distributed Inner-Product Argument.} \label{app:ipa}
A core building block of the bulletproof system is the \acf{IPA}. It allows the prover to convince a verifier that a scalar $c \in \mathbb{Z}_p$ is the correct inner-product of two vectors $a$ and $b$ where $a, b \in \mathbb{Z}_p^n$.
Bulletproofs employs a recursive argument to demonstrate the validity of an inner product relation under cryptographic commitments. In the bulletproofs setup, the \ac{IPA} provides a structured proof for the relation:
\begin{multline*}
    \{(\mathbf{g}, \mathbf{h} \in \GG^n, u, P \in \GG; \mathbf{a}, \mathbf{b} \in \ZZ_p^n)\colon P = \mathbf{g}^{\mathbf{a}} \mathbf{h}^{\mathbf{b}} \cdot u^{\langle \mathbf{a}, \mathbf{b} \rangle}\}.
\end{multline*}
Here, it is assumed that $n$ is a power of 2, ensuring the inputs align correctly; if not, zero padding may be applied to the vectors. To prove the above relation, the prover conducts $k = \log_2 n$ rounds, where in each round, the prover sends a commitment to the left half of the vector $L$ and a commitment to the right half of the vector $R$ to the verifier, receives a challenge $x$, and refines the vectors $\mathbf{a}, \mathbf{b}, \mathbf{G}, \mathbf{H}$ for the next round. This reduction process is crucial for enhancing the efficiency of the protocol. If the reduced commitment maintains the prescribed relation at any round, it implies, with overwhelming probability, that the original, un-reduced commitment also satisfies the relation. The interactive protocol is described in~\cite{bunz_bulletproofs_2018}, and by employing the Fiat-Shamir heuristic, the protocol can be made non-interactive, substituting rounds of interaction with cryptographic hashes of the commitments as we do in our implementation.

When extended to multiple provers, the protocol can be adapted to a distributed setting where each prover holds parts of the vectors $a$ and $b$. Each prover can then create shares of their parts and distribute them to all parties, resulting in each prover having shares of vectors $a$ and $b$, denoted by $\llbracket \mathbf{a} \rrbracket_i, 
\llbracket \mathbf{b} \rrbracket_i$ for $i \in [N]$ where $N$ is the number of provers. These provers collaboratively generate parts of the overall proof without revealing their individual vector shares to the verifier or each other. To achieve this, each prover runs the \ac{IPA} protocol with their respective shares in parallel and broadcasts the resulting commitments $\llbracket L \rrbracket_i$ and $\llbracket R \rrbracket_i$ to obtain the challenge $x$. At the conclusion of the protocol, each party broadcasts their partial proof, which is then combined to make the proof that the verifier can use. The non-interactive distributed protocol is summarized in \cref{fig:prove-dipa} for the prover and \cref{fig:ver-dipa} for the verifier. We note that the verification in \cref{fig:ver-dipa} is done through a single multi-exponentiation as described in~\cite[Section 3.1]{bunz_bulletproofs_2018}. This method enhances verification speed, and we have adopted it in our implementation as well.

Additionally, in \cref{fig:bp_ipa} we compare the performance of using the distributed \ac{IPA} compared to running the \ac{IPA} on revealed $\mathbf{a}$ and $\mathbf{b}$ vectors. This clearly shows that the distributed version incurs a noticeable overhead over the regular version.

% prover protocol:

\begin{figure}[ht!]
\centering
\begin{mdframed}
\small
\textbf{$\mathcal{P}_{i}$'s input:} $(\mathbf{g}, \mathbf{h} \in \GG^n, u \in \GG, 
\llbracket \mathbf{a} \rrbracket_i, \llbracket \mathbf{b} \rrbracket_i \in \ZZ_p^n)$ \\
\textbf{Output:} $\pi$

\begin{itemize}[leftmargin=*]
\item[1.] Each $\mathcal{P}_{i}$ computes: \\
    $n' \leftarrow n$ ,
    $k \leftarrow \log_2 n$ \\
    $\llbracket \mathbf{a'} \rrbracket_i \leftarrow \llbracket \mathbf{a} \rrbracket_i$ ,
    $\llbracket \mathbf{b'} \rrbracket_i \leftarrow \llbracket \mathbf{b} \rrbracket_i$ \\
    $\mathbf{g'} \leftarrow \mathbf{g}$ ,
    $\mathbf{h'} \leftarrow \mathbf{h}$ \\
    $x_0 \leftarrow 0$ \\
    % \textbf{For} $j = 1$ \textbf{to} $k$ \textbf{do}: \\
    for $j = 1, \ldots, k$ \\
    $\hspace*{1em} n' = \frac{n}{2}$ \\
    $\hspace*{1em} \llbracket c_{L,j} \rrbracket_i = \langle \llbracket \mathbf{a'}_{[\slice n']} \rrbracket_i, \llbracket \mathbf{b'}_{[n'\slice]} \rrbracket_i \rangle \in \mathbb{Z}_p$ \\
    $\hspace*{1em} \llbracket c_{R,j} \rrbracket_i = \langle \llbracket \mathbf{a'}_{[n'\slice]} \rrbracket_i, \llbracket \mathbf{b'}_{[\slice n']} \rrbracket_i \rangle \in \mathbb{Z}_p$ \\
    $\hspace*{1em} \llbracket L_j \rrbracket_i = \mathbf{g}_{[n'\slice]}^{\llbracket a'_{[\slice n']}\rrbracket_i} \mathbf{h}_{[\slice n']}^{\llbracket b'_{[n' \slice]}\rrbracket_i} u^{\llbracket c_{L,j}\rrbracket_i} \in \mathbb{G}$ \\
    $\hspace*{1em} \llbracket R_j \rrbracket_i = \mathbf{g}_{[\slice n']}^{\llbracket a'_{[n'\slice]}\rrbracket_i} \mathbf{h}_{[n'\slice]}^{\llbracket b'_{[\slice n']}\rrbracket_i} u^{\llbracket c_{R,j}\rrbracket_i} \in \mathbb{G}$ \\
    $\hspace*{1em} \mathcal{P}$: open ($L_j$, $R_j$) \\
    $\hspace*{1em} \mathcal{P} : x_j \leftarrow \mathbf{H}(x_{j-1}, L_j, R_j) \in \mathbb{Z}_p$ \\
    $\hspace*{1em} \llbracket \mathbf{a'} \rrbracket_i = \llbracket \mathbf{a'}_{[\slice n']} \rrbracket_i \cdot x_j + \llbracket \mathbf{a'}_{[n'\slice]} \rrbracket_i \cdot x_j^{-1} \in \mathbb{Z}_p^{n'}$ \\
    $\hspace*{1em} \llbracket \mathbf{b'} \rrbracket_i = \llbracket \mathbf{b'}_{[\slice n']} \rrbracket_i \cdot x_j^{-1} + \llbracket \mathbf{b'}_{[n'\slice]} \rrbracket_i \cdot x_j \in \ZZ_p^{n'}$ \\
    $\hspace*{1em} \mathbf{g'} = \mathbf{g'}_{[\slice n']}^{x_j^{-1}} \circ \mathbf{g'}_{[n'\slice]}^{x_j} \in \GG^{n'}$ \\
    $\hspace*{1em} \mathbf{h'} = \mathbf{h'}_{[\slice n']}^{x_j} \circ \mathbf{h'}_{[n'\slice]}^{x_j^{-1}} \in \GG^{n'}$\\
    % \textbf{End For}
\item[2.]  $\mathcal{P}$: open ($a', b'$)
\item[3.] \textbf{Output}: ($\mathbf{L}, \mathbf{R} \in \GG^k, a', b' \in \ZZ_p$)
\end{itemize}

\end{mdframed}
\caption{Distributed \acs{IPA} prover protocol $\Pi_{DIPA}$\pcalgostyle{.Prove}.}
\label{fig:prove-dipa}
\end{figure}

% verifier protocol
\begin{figure}[ht!]
\centering
\begin{mdframed}
\small
\textbf{$\mathcal{V}$'s input:} \\$(\mathbf{g}, \mathbf{h} \in \GG^n, u, P \in \GG, \pi = (\mathbf{L}, \mathbf{R} \in \GG^k, 
a, b \in \ZZ_p))$ \\
\textbf{Output:} $\{\mathcal{V} \text{ accepts, or } \mathcal{V} \text{ rejects}\}$

\begin{itemize}[leftmargin=*]
\item[1.] $\mathcal{V}$ computes challenges $\mathbf{x} \in \ZZ_p^k$: \\
        \hspace*{1em} $x_0 \leftarrow 0$ \\
        \hspace*{1em} for $j = 1, \ldots, k$ \\
        \hspace*{2em} $x_j \leftarrow \mathbf{H}(x_{j-1}, L_j, R_j)$
        
\item[2.] $\mathcal{V}$ computes $\mathbf{s} \in \ZZ_p^n$: \\
        \hspace*{1em}$s_i = \prod_{j=1}^{k} x_j^{b(i,j)} \quad \forall i \in [1, n]$ \\
        \hspace*{1em} where: 
        $b(i, j) =
            \begin{cases}
            1 & \text{the $j$th bit of $i-1$ is 1} \\
            -1 & \text{otherwise}
            \end{cases}$ \\
% \item[3.] $\mathcal{V}$ computes $g$ and $h$ used in the final round:\\
% % compute s and g' and h'
%     \begin{align*}
%         g &= \prod_{i=1}^{n} g_i^{s_i} \in \GG \\
%         h &= \prod_{i=1}^{n} h_i^{1/s_i} \in \GG
%     \end{align*}
\item[3.] $\mathcal{V}$ checks: \\
% check P 
    $\mathbf{g}^{a \cdot \mathbf{s}} \cdot \mathbf{h}^{b \cdot \mathbf{s}^{-1}} \cdot u^{a \cdot b} \stackrel{?}{=} P \cdot \prod_{j=1}^{k} L_j^{(x_j^2)} \cdot R_j^{(x_j^{-2})}$ \\
    - If yes, $\mathcal{V}$ accepts; otherwise $\mathcal{V}$ rejects.
\end{itemize}

\end{mdframed}
\caption{Distributed \acs{IPA} verifier protocol $\Pi_{DIPA}$\pcalgostyle{.Verify}.}
\label{fig:ver-dipa}
\end{figure}

% verifier protocol
\subsection{Bulletproofs Verifier Protocol}\label{app:bp-verify}
In \cref{fig:bp-verify}, we present the non-interactive version of the Bulletproof verification algorithm, adapted for the commit-and-prove setting. At a high level, this algorithm mirrors the standard verification protocol with the addition of \cplink verification. Besides the proof and public parameters, the verification protocol requires $\hat{V}$, an external commitment to the same witnesses in $\mathbf{V}$, and $\mathbf{vk}_{\cplink}$, the verification key used for \cplink verification. This ensures that the witnesses in $\mathbf{V}$ satisfy the arithmetic circuit and are linked to the external commitment $\hat{V}$. Additionally, this Bulletproof verification uses \ac{IPA} to reduce communication costs, as detailed in \cref{fig:ver-dipa}.

\begin{figure}[ht!]
\centering
\begin{mdframed}
\small
\textbf{$\mathcal{V}$'s input:} \\ 
% $(g, h, \in \GG, \mathbf{g}, \mathbf{h} \in \GG^n, \mathbf{W}_L, \mathbf{W}_R, \mathbf{W}_O, \mathbf{W}_V, \mathbf{c}, \pi, \mathbf{V} \in \GG^m, \pi_{\cplink}, \hat{V} \in \GG), \mathbf{vk}_{\cplink}$ \\
\[
\begin{pmatrix}
g, h \in \mathbb{G}, \mathbf{g}, \mathbf{h} \in \mathbb{G}^n, 
\mathbf{c} \in \mathbb{Z}_p^Q, \mathbf{V} \in \GG^m, \hat{V} \in \GG, \\
\mathbf{W}_L, \mathbf{W}_R, \mathbf{W}_O \in \mathbb{Z}_p^{Q \times n}, \mathbf{W}_V \in \mathbb{Z}_p^{Q \times m}, \pi_{\cplink}, \mathbf{vk}_{\cplink}, \\
\pi = (A_I, A_O, S, T_1, T_3, T_4, T_5, T_6, \tau_x, \mu, \hat{t}, \pi_{IPA})
\end{pmatrix}
\]
\textbf{Output:} $\{\mathcal{V} \text{ accepts,} \mathcal{V} \text{ rejects}\}$ \\
\begin{itemize}[leftmargin=*]
    \item[1.] $\mathcal{V}$ computes challenges from $\pi\colon \{ y, z, x, x_u \}$ \\
    \item[2.] $\mathcal{V}$ computes and checks:
        \begin{align*}
        \mathbf{y}^n &= (1, y, y^2, \ldots, y^{n-1}) \in \ZZ_p^n \\
        \mathbf{z}^{Q+1}_{[1\slice]} &= (z, z^2, \ldots, z^Q) \in \ZZ_p^Q \\
        h'_i &= h_i^{y^{-i+1}} \quad \forall i \in [1, n] \\
        W_L &= \mathbf{h'}^{\mathbf{z}^{Q+1}_{[1:]} \cdot \mathbf{W}_L} \\
        W_R &= \mathbf{g}^{\mathbf{y}^{-n} \circ \left( \mathbf{z}^{Q+1}_{[1\slice]} \cdot \mathbf{W}_R \right)} \\
        W_O &= \mathbf{h'}^{\mathbf{z}^{Q+1}_{[1\slice]} \cdot \mathbf{W}_O} \\
        P &= A^x_I \cdot A^{(x^2)}_O \cdot \mathbf{h'}^{-\mathbf{y}^n} \cdot W^x_L \cdot W^x_R \cdot W_O \cdot S^{(x^3)} \\
        P' &= P \cdot h^{-\mu} \cdot g^{x_u \cdot \hat{t}}\\
        \{0, 1\} &\leftarrow \Pi_{IPA}.\text{Verify}(\mathbf{g}, \mathbf{h'}, g^{x_u}, P', \pi_{IPA})
        \end{align*}
    \item[3.] $\mathcal{V}$ computes and checks:
        \begin{align*}
            V' &= \sum_{j=0}^{m} V_j\\
            \{0, 1\} &\leftarrow \cplink.\mathsf{Verify}(\mathbf{vk}_{\cplink}, V', \hat{V}, \pi_{\cplink})
        \end{align*}
    \item[4.] If all checks succeed, $\mathcal{V}$ accepts, else $\mathcal{V}$ rejects.
\end{itemize}
\end{mdframed}
\caption{Protocol $\Pi_{DBP}\pcalgostyle{.Verify}$ for verifying bulletproofs in the commit-and-prove setting (with \cplink).}
\label{fig:bp-verify}
\end{figure}

\subsection{Additional Experimental Evaluations} \label{apx-eval}
\textbf{\cplink overhead.}
In this experiment, we measured the overhead of \cplink, a technique, as described previously, that allows linking the commitment to the witnesses in both LegoGro16 and Bulletproofs to an external commitment. The results provide insight into the cost associated with linking proofs in a distributed setting compared to a single prover setting. While \cplink does come with a cost (overhead), it is certainly far more efficient than linking the commitments inside the circuit. \cref{fig:cplink} shows the performance of \cplink for an increasing number of provers where each prover contributes a witness to the commitment.

\begin{figure}[htbp]
    \centering
    \begin{minipage}{0.45\textwidth}
        \centering
        \includegraphics[width=\textwidth]{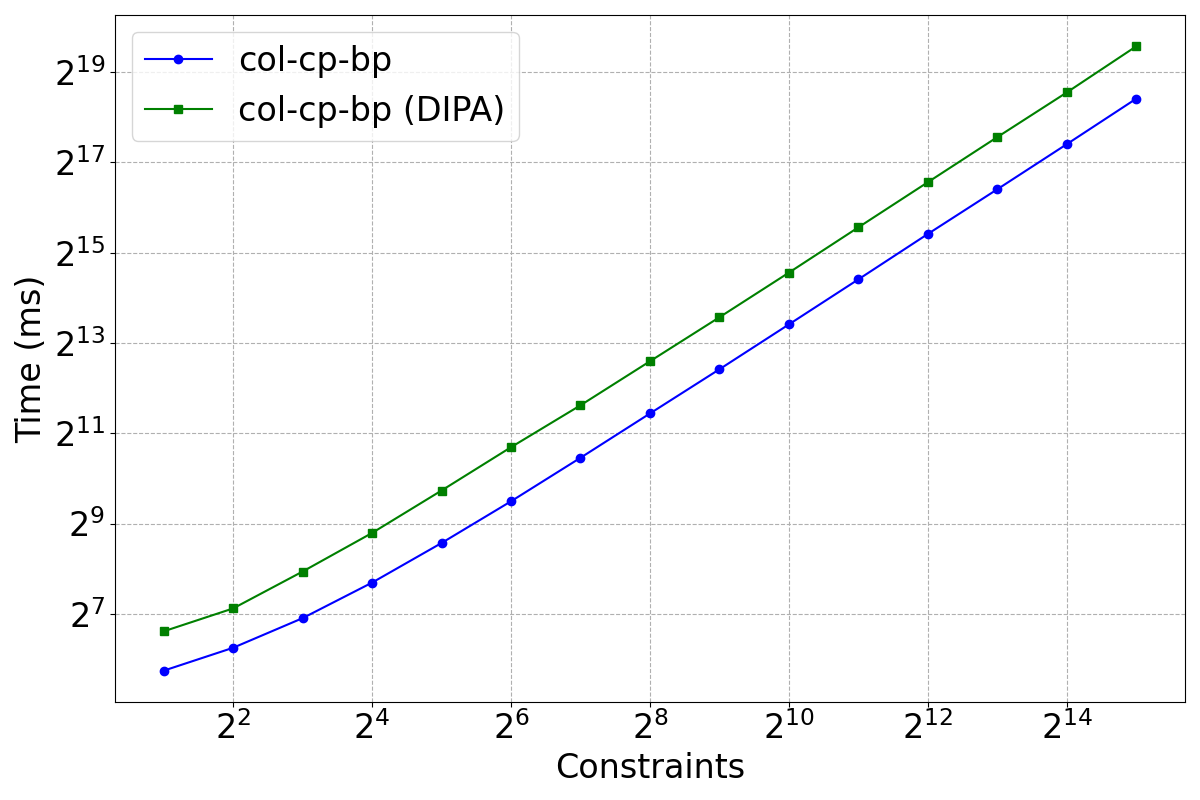}
        \caption{Runtime of a single prover for: (1) \ac{IPA} on opened $\mathbf{a}$ and $\mathbf{b}$ input vectors (col-cp-bp); (2) on shared input (col-cp-bp(DIPA)). The number of provers is fixed at $2$.}
        \label{fig:bp_ipa}
    \end{minipage}
    \hfill
    \begin{minipage}{0.45\textwidth}
        \centering
        \includegraphics[width=\textwidth]{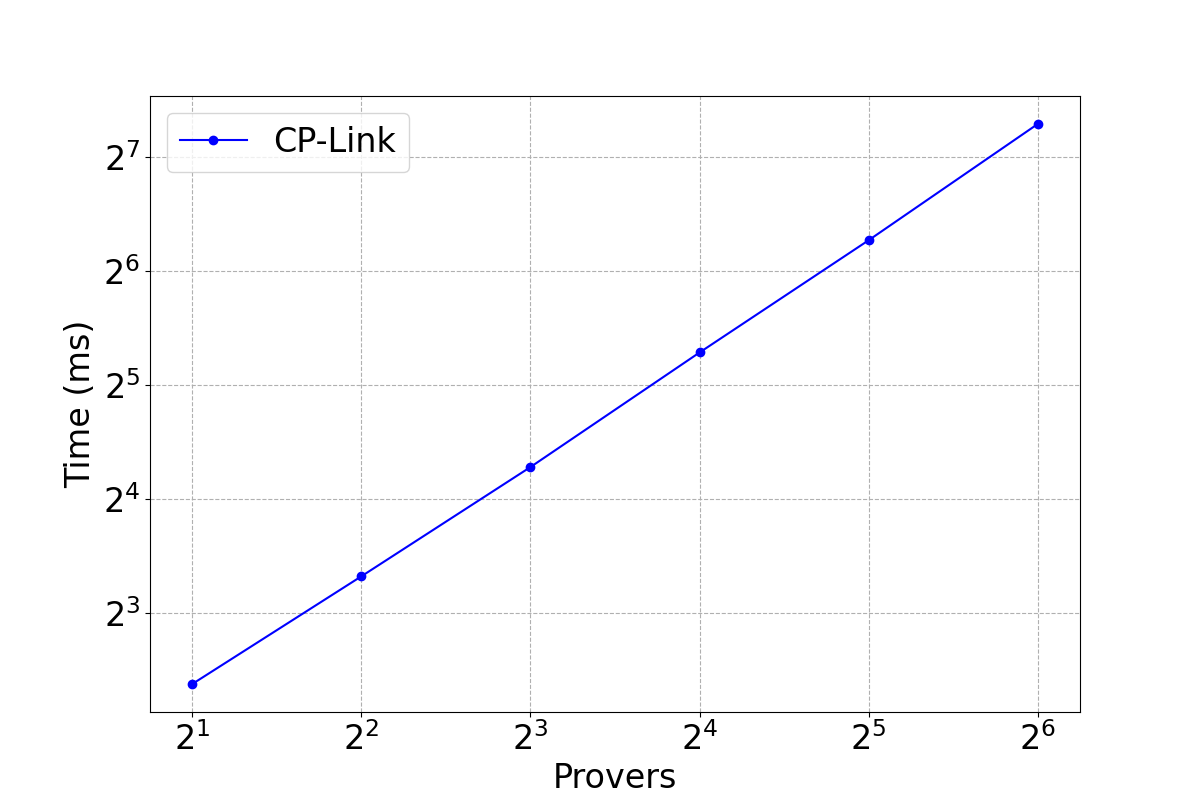}
        \caption{Runtime per prover for \cplink with a varying number of provers.}
        \label{fig:cplink}
    \end{minipage}
\end{figure}

% \begin{figure}[htbp]
% \centering
% \includegraphics[width=0.40\textwidth]{img/bp_ipa.png}
% \caption{Runtime of a single prover for: (1) \ac{IPA} on opened $\mathbf{a}$ and $\mathbf{b}$ input vectors (col-bp); (2) on shared input (col-bp(dist-IPA)). The number of provers is fixed to $2$.}
% \label{fig:bp_ipa}
% \end{figure}

% \begin{figure}[htbp]
% \centering
% \includegraphics[width=0.40\textwidth]{img/cplink.png}
% \caption{Runtime per prover for \cplink with a varying number of provers.}
% \label{fig:cplink}
% \end{figure}

\end{document}